\documentclass[11pt]{article}

\usepackage[paperwidth=8.5in,paperheight=12in, margin=1in]{geometry}
\usepackage{amsmath,amssymb,amsthm,amsfonts,cases}
\usepackage{mathrsfs}
\usepackage{bm}
\usepackage{caption}
\usepackage{subcaption}
\usepackage{multirow}
\usepackage{xurl}
\usepackage[dvipsnames]{xcolor}

\usepackage{hyperref} 
\hypersetup{
	colorlinks   = true,
	citecolor    = blue,
	linkcolor    = OrangeRed,
	urlcolor     = blue
}

\usepackage{enumitem}
\usepackage[authoryear]{natbib}
\usepackage{mathtools}
\mathtoolsset{showonlyrefs=true}
\numberwithin{equation}{section}

\usepackage{booktabs}
\usepackage{pgfplots}
\usetikzlibrary{math}

\usepackage{comment}
\usepackage{floatrow}
\usepackage{authblk}

\newtheorem{problem}{Problem}
\newtheorem{assumption}{Assumption}
\newtheorem{theorem}{Theorem}[section]
\newtheorem{proposition}[theorem]{Proposition}
\newtheorem{corollary}[theorem]{Corollary}
\newtheorem{lemma}[theorem]{Lemma}
\newtheorem{remark}{Remark}[section]
\newtheorem{example}{Example}[section]

\newcommand{\Eb}{\mathbb{E}}
\newcommand{\Ib}{\mathbb{I}}

\newcommand{\Pb}{\mathbb{P}}
\newcommand{\Rb}{\mathbb{R}}

\newcommand{\Ac}{\mathcal{A}}

\newcommand{\AcIC}{\Ac_{\mathrm{IC}}}

\newcommand{\Fc}{\mathcal{F}}
\newcommand{\Ic}{\mathcal{I}}

\newcommand{\Sc}{\mathcal{S}}

\newcommand{\dd}{\mathrm{d}}

\newcommand{\pf}{\pi_{f}}

\newcommand{\tI}{\widetilde{I}}

\DeclareMathOperator*{\argsup}{argsup}

\begin{document}

\title{Optimal Reinsurance under Endogenous Default and Background Risk}

\author{Zongxia Liang\thanks{Department of Mathematical Sciences, Tsinghua University, China. Email: \texttt{liangzongxia@tsinghua.edu.cn}} \quad 
	 Zhaojie Ren\thanks{Department of Statistics and Data Science, The Chinese University of Hong Kong, Hong Kong, China. Email: \texttt{rzj20@mails.tsinghua.edu.cn}} 
	 \quad 
	 Bin Zou\thanks{Department of Mathematics, University of Connecticut, USA. Email: \texttt{bin.zou@uconn.edu}}}	
 
\date{This version: \today \\
Forthcoming in \emph{ASTIN Bulletin}}

\maketitle

\begin{abstract}
	\noindent
This paper studies an optimal reinsurance problem for a utility-maximizing insurer, subject to the reinsurer's endogenous default and background risk. An endogenous default occurs when the insurer's contractual indemnity  exceeds the reinsurer's available reserve, which is random due to the background risk.
We obtain an analytical solution to the optimal contract for two types of reinsurance contracts, differentiated by whether their indemnity functions depend on the reinsurer's background risk.
The results shed light on the joint effect of the reinsurer's default and background risk on the insurer's reinsurance demand.  
\end{abstract}


\section{Introduction}
\label{sec:intro}

\subsection{Background and Motivation}
\label{sub:motivation}
A typical reinsurance contract involves two parties, an insurer (buyer of reinsurance) and a reinsurer (seller of reinsurance), and it is uniquely identified by its indemnity function $I$, which specifies the indemnification amount $I(x)$ to the insurer in case of a covered loss with size $x>0$.
The design of optimal reinsurance contracts aims to identify the optimal indemnity function  under a chosen criterion and has long been a pivotal topic in insurance economics and actuarial science (see  \citet{arrow1963uncertainty}).
The goal of this paper is to derive an optimal reinsurance contract for a utility-maximizing insurer, subject to the reinsurer's endogenous default and background risk. 

The classical literature on optimal (re)insurance  makes an \emph{implicit} assumption that the seller of contract can always fulfill its promised contract payment, specified by $I$, to the buyer. However, this assumption is challenged in real-world reinsurance markets because many factors, such as catastrophic climate events or systemic risk in the financial markets, could lead to the partial (or even no) payment on covered losses from the seller (\citet{cummins2003optimal}). The failure to pay the contractual indemnity $I$ from the seller is a particular case of \emph{contract nonperformance} (\citet{doherty1990rational}) and, as one would expect, has a significant impact on the buyer's (re)insurance demand (\citet{peter2020you}).

There are two main contributors to the failure of reinsurance companies. The first one is a catastrophic event that results in a large realization of insured losses. In a recent example,  Hurricanes Helene and Milton hit Florida within a two-week period, together causing insured losses estimated between \$30 billion and \$50 billion, and put several small, Florida-based reinsurance companies
under huge financial stress. 
The second contributor comes from the reinsurer's \emph{background risk}, broadly defined as random sources that impact the reinsurer's reserve for settling claims. One particular source of background risk is the reinsurer's exposure to financial risks, but it can also include, for instance, geopolitical and social risks, and even pandemic risk (\citet{IAIS2012}).
To understand the nonperformance of reinsurance contracts in theory, note that contractual indemnity $I$ is often an increasing function of the insurer's loss, and there exists a threshold $\bar{x}$, potentially depending on the reinsurer's background risk, such that for all $x > \bar{x}$, the contractual indemnity $I(x)$ exceeds the reinsurer's available reserve, resulting in an \emph{endogenous} default. 
To summarize, both empirical and theoretical evidence motivate us to incorporate the reinsurer's endogenous default and background risk in the study of optimal reinsurance.

\subsection{Summary and Contributions of the Paper}

We study an optimal reinsurance problem in a one-period model for a utility-maximizing insurer who is exposed to an insurable loss $X$. For a reinsurance contract with indemnity function $I$, denote $\mathcal{I}$ as the indemnity payment (traditionally $\mathcal{I}:=I(X)$); the reinsurer charges a premium \( \pi(I) = (1 + \eta)\Eb[\Ic] \) under the expected-value principle, with loading factor \( \eta \geq 0 \).
The reinsurer's terminal reserve is $R = (S +  \pi(I))^+$, in which $S$ is a random variable and captures the reinsurer's background risk. With this setting, the reinsurer's initial reserve \( r \) at time $0$ evolves stochastically to \( S \) at time 1; as such, \( S - r \) is the background risk, but mathematically, it is equivalent to directly call \(S\) the reinsurer's background risk. 
As argued in Section \ref{sub:motivation}, we consider the reinsurer's endogenous default and define it as an event whenever the indemnity payment $\Ic$ exceeds the reinsurer's reserve $R$. As such, for a contract $I$, the indemnity payment received by the insurer changes from ``\emph{contractual} amount'' $\Ic$  to ``\emph{actual} amount'' $\Ib = \Ic \cdot \mathbf{1}_{\{ \Ic \leq R \}} + \tau R \cdot \mathbf{1}_{\{ \Ic > R \}}$, in which $\tau \in [0,1]$ is the recovery rate in case of an endogenous default.  The insurer is aware that its reinsurance contract may be nonperforming
and seeks an optimal contract $I^*$ to maximize the  expected utility of its terminal wealth.

In Section \ref{sec:main}, the insurer is allowed to choose contracts in the form of $I := I (x, s)$;
i.e., the indemnity $I$ can be a function of both the insurer's loss size $x$ and the reinsurer's background risk (random reserve) level $s$.
We  follow a two-step approach to obtain the insurer's \emph{globally} optimal contract $I^*$. In the first step, we fix the contract premium at a given level $a$ 
and obtain the \emph{locally} optimal contract $I_a^*$ (see Theorem \ref{thm:step_one}). In the second step, we optimize over all feasible $a$ and identify the optimal premium level $a^*$. The globally optimal contract is then given by $I^* = I^*_{a^*}$ (see Theorem \ref{thm:1_ass}). We  obtain $I^*$ in a semiclosed form and show that it is a single deductible reinsurance with a policy limit; in addition, the reinsurer will never default under contract $I^*$. 
These results are obtained over the \emph{largest} possible set of  admissible indemnities (see Remark \ref{rem:ind}) and require only mild conditions on $X$ and $S$ (Assumption \ref{ass}); to the best of our knowledge, they are new to the optimal reinsurance literature.  
We also  derive \emph{analytical} results on the comparative statics of the optimal contract, including its deductible, policy limit, and premium (see Proposition \ref{prop:a_p}), while the existing literature obtains limited results, mostly relying on numerical analysis.  In particular, we show that the optimal deductible decreases with respect to the reinsurer's reserve and the insurer's risk aversion, but increases with respect to the insurer's initial wealth.

In Section  \ref{sec:ext}, the insurer can only choose contracts in the form of $I := I(x)$,  depending only on the loss size $x$, which form a subclass of the contracts considered in Section \ref{sec:main}. However, the resulting optimal reinsurance problem is challenging to solve, and the two-approach approach in Section \ref{sec:main} \emph{cannot} be applied here. To our awareness, such an optimal reinsurance problem has not been studied previously in the literature. Imposing the incentive compatible (IC) condition on $I$ and assuming that the discrete $S$ is independent of $X$, we characterize the (locally) optimal contract $I_S^* := I_S^*(x)$ in a parametric form (see Theorem \ref{thm:N}). We show that $I_S^*$ is a \emph{multi-layer} reinsurance contract, and each layer involves a deductible and a policy limit, both depending on a free parameter $l_i$.  With additional conditions (e.g., $I_S^*$ has two layers), we can further obtain the optimal parameters, $l_1^*$ and $l_2^*$, and fully identify the optimal contract $I_S^*$ (see Proposition \ref{prop:N=2}). 
We remark that endogenous default may occur under $I_S^*$ in Section \ref{sec:ext}, but $I^*$ obtained in Section \ref{sec:main} is a \emph{default-free} contract. This striking difference highlights the impact of the reinsurer's background risk and the choice of contracts on decision making.

\subsection{Related Literature}

The seminal work of \citet{arrow1963uncertainty} considers a one-period insurance model without the counterparty default and background risk, and it shows that the optimal contract is a deductible insurance for a utility-maximizing insured under the expected-value premium principle. A significant body of the literature on optimal (re)insurance aims to extend Arrow's model by exploring different optimization criteria or premium principles  (see, e.g., \citet{bernard2015optimal} for rank-dependent utility (RDU) and  \citet{birghila2023optimal} for maximin expected utility;  \citet{asimit2013optimal} for distortion premium principles); we refer to \citet{gollier2013economics} and  \citet{cai2020optimal}
for an overview of the research on optimal (re)insurance. 
This paper, on the other hand, incorporates the counterparty default risk and background risk into Arrow's model and studies the insurer's optimal reinsurance problem accordingly. As such, we focus on reviewing related works that consider default or background risk.

In terms of modeling default risk, there are two main approaches. The first approach models the seller's default as an \emph{exogenous} event, often by a binary random variable independent of the contract. \citet{doherty1990rational} follow this approach in their model and are among the earliest to study insurance demand under default risk (restricting to proportional contracts). In addition, they assume that the default probability is known to both parties, and when default occurs, \emph{no} indemnity payment is made to the buyer (i.e., the recovery rate is $\tau = 0$).
As they write, ``A nonzero probability of default renders most of the standard insurance results invalid'' (p.244), and their findings demonstrate the profound impact of default risk on insurance demand. Subsequent studies extend \citet{doherty1990rational} by considering heterogeneous beliefs (\citet{cummins2003optimal}), ambiguity preferences (\citet{peter2020you}), hedging strategies (\citet{reichel2022optimal,chi2023optimal}), and distortion risk measures (\citet{yong2024optimal}), among many others.

The second approach, adopted by this paper, models the seller's default as an \emph{endogenous} event which occurs when the seller's reserve is less than the contractual indemnity. This approach is arguably more realistic than the first one but leads to a more challenging optimal (re)insurance problem, which may explain why there is only limited study under endogenous default. 
In an early attempt,  \citet{biffis2012optimal} study a Pareto optimal insurance problem under endogenous default (and background risk) for a risk-averse insured and a risk-neutral insurer.
They derive several properties of the optimal insurance, should it exist, but fail to find an analytical solution to the optimal contract.
\citet{asimit2013optimal} consider a more concrete setup (but without background risk) in which the reinsurer's reserve is based on the Value-at-Risk (VaR) rule, and reinsurance contracts are priced by the distortion premium principle. 
When the insurer aims to minimize its risk (measured by either VaR or a distortion risk measure), they obtain the optimal contract in (semi)closed form. \citet{cai2014optimal} conduct a similar study as \citet{asimit2013optimal} but assume the expected-value premium principle and no bankruptcy costs ($\tau = 1$); they find an analytical solution for the optimal contract that maximizes the insurer's expected utility or minimizes the VaR. Note that neither \citet{asimit2013optimal} nor \citet{cai2014optimal} consider the reinsurer's background risk.
More recently, \citet{chen2024bowley} seek a Bowley solution to a reinsurance game in the presence of endogenous default risk.

Next, we discuss existing research on optimal (re)insurance problems that considers background risk. 
The surplus of a company (e.g., the seller of (re)insurance policies) or the wealth of an individual insured  is exposed to various risks. Here, we broadly define \emph{background risk} as random risks that impact the seller's or buyer's surplus, and they are either uninsurable or not insured. 
Recall that the initial wealth of the (re)insurance buyer is a constant in the classical models (see \citet{borch1962equilibrium} and \citet{arrow1963uncertainty}); much of the effort is devoted to incorporating the \emph{buyer's} (referring to the insurer in a reinsurance context or the insured in an insurance setting) background risk into the model. \citet{doherty1983optimal} study optimal deductible insurance when the insured's initial wealth is random (due to background risk). \citet{mayers1983interdependence} consider a particular case in which the insured's background risk comes from the financial market, and they study the optimal demand for proportional insurance and financial assets.  \citet{chi2020optimal} and \citet{chen2024optimal} explore the optimal insurance problem under a general dependence structure between the insured's insurable and background risks. However, the \emph{seller's} (corresponding to the reinsurer of reinsurance  or the insurer of insurance contracts) background risk is at least as important as the buyer's background risk, if not more so; however, related research is quite limited. The model in \citet{biffis2012optimal} is a rare example and considers \emph{both} the buyer's and seller's background risk; but they do not obtain an explicit optimal contract. \citet{filipovic2015optimal} consider the seller's background risk and allow it to invest in a risky asset, while \citet{boonen2019equilibrium} extends their work to an equilibrium model with multiple insureds (buyers).

\section{Model}
\label{sec:model}

We consider a one-period model and study the optimal reinsurance problem for an insurer.
We fix a complete probability space $(\Omega, \Fc, \Pb)$ and denote by $\Eb[\cdot]$ the expectation taken under $\Pb$. All (in)equalities involving random variables should be understood in the $\Pb$ almost surely sense.
The insurer's aggregate loss (or a portfolio of risks) is modeled by a nonnegative $\Fc$-measurable random variable $X$, and we assume that $X$ is bounded from above, $0<M:=\mathrm{ess \, sup}\,X<\infty$. This assumption is common in the literature (see, e.g., \citet{chi2023optimal}), and it is rather mild, because $M$ can be arbitrarily large. The main results can be readily extended to the case of $M = \infty$ under certain integrability conditions (see Online Companion for full details on this extension), while $M=0$ trivializes the problem. 
To mitigate the risk exposure, the insurer purchases reinsurance from the reinsurance market, and we denote the \emph{indemnity} function of a reinsurance contract by $I$ and the corresponding stochastic indemnity payment  by $\Ic$ (traditionally $\mathcal{I}:=I(X)$). 
Given the one-to-one relation between a reinsurance contract and its indemnity function $I$, we often call $I$ a contract.
We assume that the reinsurer applies the expected-value premium principle to calculate the contract premium by 
\begin{align}
	\label{eq:pi}
    \pi(I) = (1 + \eta) \, \Eb[ I(X) ], 
\end{align}
in which $\eta \ge 0$ is the premium loading factor. Using the actual indemnity $\Ib$ to determine the premium leads to a ``loop'' issue, because the actual indemnity and premium are intertwined through the reinsurer's random reserve. This explains the standard form of the premium in \eqref{eq:pi}; see \citet{cai2014optimal} for the same premium principle when the seller’s default is present. Because the reinsurer may default on indemnity payment, the loading factor $\eta$ should reflect this  default risk.\footnote{Let $\bar{\eta} > 0$ denote the loading factor charged by a default-free reinsurer; then, $\eta < \bar{\eta}$. For instance, assuming that the default probability, $p_{de} \in (0,1)$, can be estimated, we may set $\eta$ so that $(1 + \eta) = (1 + \bar{\eta}) ( 1- p_{de})$.}

Once a contract $I$ is chosen by the insurer, the reinsurer's terminal available reserve, $R$, for settling claims is given by
\begin{align}
	\label{eq:R}
	R = (S+\pi(I))^+, 
\end{align}
in which $S$ is an $\Fc$-measurable random variable capturing the reinsurer's background risk, and $y^+ : = \max\{0, y\}$ for all $y \in \Rb$. The rationale is that the reinsurer sets aside an  initial reserve of amount $r$ at time 0, which consists of both ``cash'' and ``risky assets,''
and thus its value at time 1, $S$, is random. In that regard, one can also call $S$ the reinsurer's \emph{random reserve} (excluding premium), and the difference $S - r$ models the reinsurer's additive background risk. But for notational simplicity, we call $S$ the reinsurer's background risk in the sequel.

An \emph{endogenous} default from the reinsurer occurs if (and only if) $R< \Ic$, namely when the reinsurer's terminal reserve falls short of the contractual indemnity.  We assume that in case of an endogenous default, the insurer receives a fraction, $\tau \in [0,1]$, of the reinsurer's available reserve.\footnote{We assume that the so-called recovered rate $\tau$ is a constant in this paper. However, in reality, such a rate may be random, depending on additional risk factors at time 1; see \cite{bernard2012impact} and \cite{li2018optimal} for this setup. Our results in Section \ref{sec:main} can be readily extended to any random $\tau \in [0,1]$ without further assumptions. The findings in Section \ref{sec:ext}, however, require the additional assumption that $\tau$ is independent of the background risk $S$ to remain valid.} 
Note that our setup allows for large negative shocks of $S$ that result in a negative value of $S + \pi(I)$, which is why we take the positive part in \eqref{eq:R}.
Mathematically, we introduce a binary variable, $D:= D(I)$, to track the reinsurer's solvency status by
\begin{align}
	\label{eq:D}
	D = \begin{cases}
		1 \text{ (default)}, & \text{ if } R < \Ic, \\
		0 \text{ (solvent)}, & \text{ if } R \ge \Ic.
	\end{cases}
\end{align}
$D$ depends on the insurer's loss $X$ and contract $I$,  and also the reinsurer's background risk $S$. 

Because reinsurance contracts may be nonperforming due to the reinsurer's endogenous default, the \emph{actual} indemnity payment, $\Ib$, that the insurer receives from a contract $I$ is given by 
\begin{align}
	\label{eq:Ib}
	\Ib(I) = \Ic \cdot \mathbf{1}_{\{ D = 0 \} } + \tau R \cdot \mathbf{1}_{ \{ D = 1 \}}, 
\end{align}
in which $\mathbf{1}$ denotes an indicator function, and $R$ and $D$ are defined by \eqref{eq:R} and \eqref{eq:D}, respectively. Therefore, for a chosen contract $I$, the insurer's terminal wealth, $W$, equals 
\begin{align}
	\label{eq:W}
	W(I) = w - X - \pi(I) + \Ib(I), 
\end{align}
in which $w$ is the insurer's initial wealth, $\pi(I)$ is given by \eqref{eq:pi}, and $\Ib$ is the actual indemnity defined in \eqref{eq:Ib}. Following the classical literature on (re)insurance economics (e.g., \citet{arrow1963uncertainty} and \citet{mossin1968aspects}), we assume that the insurer's preferences  are modeled by the expected utility theory with a twice differentiable utility function $u$ that is strictly increasing and strictly concave (i.e., $u'>0$ and $u''<0$).  We formulate the main problem of the paper as follows.

\begin{problem}\label{prob:main}
	The insurer seeks an optimal reinsurance contract $I^* $ to maximize the expected utility of its terminal wealth under the reinsurer's endogenous default and background risk, i.e., $$I^* = \argsup_{I \in \widetilde \Ac} \, \Eb[u(W(I))],$$ in which $\widetilde{\Ac}$ is the admissible set to be specified later, and  $W(I)$ is given by \eqref{eq:W}.	
\end{problem}

\begin{remark}
	In Problem \ref{prob:main}, the premium of a reinsurance contract $\pi(I)$ is computed by the expected-value premium principle in \eqref{eq:pi}. As is standard in the study of optimal (re)insurance, the form of the optimal contract often depends on the chosen premium principle. For example, \cite{chen2019stochastic} show that the optimal contract (in their setting) is of deductible form under the expected-value premium principle, but it is of proportional form under the variance premium principle. Among the results obtained in this paper, Propositions \ref{prop:extre} and \ref{prop:a_bar} remain valid under mean-variance and distortion premium principles, but other results likely are limited to the expected-value premium principle. It is an interesting direction to revisit Problem \ref{prob:main} under different premium principles.
\end{remark}

By definition, the reinsurer defaults if $(S + \pi(I) )^+< \Ic$, suggesting that the reinsurer's background risk $S$ plays a key role in its solvency.  Because  $I\ge 0$, we first consider an extreme scenario of $S\le 0$ and show that the optimal strategy is to \emph{not} purchase any reinsurance.  

\begin{proposition}
	\label{prop:extre}
If $S \le 0$ almost surely (i.e., $\mathbb{P}(S \le 0) = 1$), then the optimal strategy to Problem \ref{prob:main} is no reinsurance with $I^* \equiv 0$.
\end{proposition}

\begin{proof} 
	For any non-negative indemnity function $I\ge 0$ (thus $\pi(I) \ge 0$), we obtain 
	\begin{align}
		W(I)&=w - X - \pi(I) + \Ic \cdot \mathbf{1}_{\{\Ic \le (S+\pi(I))^+\}} + \tau (S+\pi(I))^+  \cdot \mathbf{1}_{\{\Ic >(S+\pi(I))^+\}} \\
		&\le w - X - \pi(I) + (S+\pi(I))^+ \\
		&\le w - X - \pi(I) + \pi(I)=w-X=W(0).\label{eq:prop:extre}
	\end{align}
	Because $u' >0$, $\Eb[u(W(I))] \le \Eb[u(W(0))]$, and the result follows. 
\end{proof}

The result in Proposition \ref{prop:extre} already showcases the important impact of counterparty default on decision making. 
We know from, for instance, \citet{arrow1963uncertainty} that, if the reinsurer's default risk is \emph{ignored}, the optimal contract is a deductible reinsurance. 
However, as Proposition \ref{prop:extre} shows, when the insurer is aware of the counterparty default risk, and the reinsurer's reserve is nonpositive, the optimal decision is to \emph{not} purchase any reinsurance but to fully rely on self-insurance. Because the case of $S \le 0$ is solved by Proposition \ref{prop:extre}, and it is likely unrealistic in practice, we study Problem \ref{prob:main} under the standing assumption that $\Pb(S > 0) > 0$ in the rest of the paper. 

Because there are two random sources, $X$ and $S$, in the model, we consider two types of indemnity functions in the analysis:
\begin{enumerate}
	\item Loss- and background-risk-dependent indemnities $I : = I(x, s)$ (denoting $\Ic : = I(X, S)$).
	
	\item Loss-dependent indemnities $I:=I(x)$ (denoting $\Ic:=I(X)$).
\end{enumerate}
Reinsurance contracts with indemnities in the form of $I(x, s)$ are examples of the so-called \emph{randomized} contracts in the reinsurance literature. Similar randomized (re)insurance contracts are  considered in \citet{albrecher2019randomized} with an independent Bernoulli random variable and  in \citet{asimit2021risk} within multiple indemnity environments. Moreover, we also consider the second type of contracts with indemnity $I(x)$  depending only on the insurer's own loss size $x$, but independent of the reinsurer's background risk level $s$. Note that the second type is the more conventional contract form and constitutes a subclass of the first type $I(x, s)$, and the two types coincide when there is no background risk ($S$ reduces to a constant). In either type of contracts, the reinsurer's endogenous default is induced by large losses from the insurer or negative shocks from its background risk. 

We discuss the actual indemnity payments of these two contract types as follows.  
\begin{enumerate}
	\item \( \Ic = I(X, S) \): The contractual payment \( \Ic \) depends on the realized values of both \( X \) and \( S \). Consequently, the actual indemnity payment \( \Ib \) in \eqref{eq:Ib} is contingent on \( S \) even when no default occurs, necessitating regulatory disclosure of the reinsurer's reserve \( S \) at time 1 (e.g., audited financial statements or quarterly reports) to ensure execution of the contract.  
	
	\item \( \Ic = I(X) \): The contractual payment \( \Ic \) depends only on \( X \), and thus the actual payment \( \Ib \) in \eqref{eq:Ib} is independent of \( S \) unless default occurs. Because default is readily observable, this contract aligns with conventional regulatory frameworks that focus on solvency verification rather than continuous reserve monitoring.
\end{enumerate}

The above distinction in payment mechanism underscores the fact that the first contract type—while offering enhanced risk sharing—requires stricter regulations and transparency, whereas the second type remains prevalent in practice due to their operational simplicity. 

\section{Optimal Loss- and Background-Risk-Dependent Indemnities}
\label{sec:main}

In this section, we consider loss- and background-risk-dependent indemnities in the form of $ I=I(x, s)$; i.e., the insurer is allowed to choose reinsurance contracts that depend on both its own loss size $x$ and the reinsurer's background risk level $s$. We impose the minimum condition--the so-called ``\emph{principle of indemnity}''--on the indemnity function $I$, which leads to the following  \emph{admissible} set $\widetilde{\Ac} = \Ac$:
\begin{align}
	\label{eq:Ac}
	\Ac :=  \{I: [0,M]\times\Rb \mapsto \Rb_+ \, | \, 0 \le I(x,s) \le x \text{ for all } x \ge 0 \text{ and all } s\in \Rb\}.
\end{align}

For every $I \in \Ac$, noting $R = (S+\pi(I))^+$,  the insurer's terminal wealth $W$ is given by 
\begin{align}
	\label{eq:W_no}
	W(I)
	= w - X - \pi(I) + I(X,S) \cdot \mathbf{1}_{ \{R \ge I(X,S)\} } + \tau R \cdot \mathbf{1}_{ \{ R < I(X,S) \} }.
\end{align} 

We state the first concrete version of Problem \ref{prob:main} as follows.

\begin{problem}
	\label{prob:no}
		The insurer seeks an optimal loss- and background-risk-dependent reinsurance contract $I^*:=I^*(x, s) \in \Ac$ to maximize the expected utility of its terminal wealth under the reinsurer's endogenous default and background risk, i.e., $$I^* = \argsup_{I \in \Ac} \, \Eb[u(W(I))],$$
	in which the admissible set $\Ac$ is defined in \eqref{eq:Ac}, and $W(I)$ is given by \eqref{eq:W_no}.  
\end{problem}

\begin{remark}
	\label{rem:ind}
	The admissible set $\Ac$ in \eqref{eq:Ac} can be seen as an extension to the one, $\{I: [0,M] \mapsto \Rb_+ \, | \, 0 \le I(x) \le x   \text{ for all } x \ge 0 \}$, used in the classical literature (see \citet{arrow1963uncertainty} and \citet{mossin1968aspects}) and is likely the largest admissible set one can consider for a meaningful optimal (re)insurance problem. Indeed,  related research often imposes further conditions on admissible indemnities. A prime example is the so-called \emph{incentive compatibility} (IC) condition, which reads as $0 \le I(x) - I(x') \le x - x'$ for all $x \ge x' \ge 0$, and is imposed to rule out certain \emph{ex post} moral hazard (see \citet{huberman1983optimal}). Under the IC condition, both the indemnity function $I$ and the retained loss function $x - I$ are nondecreasing and satisfy the 1-Lipschitz condition (implying that they are differentiable almost everywhere with the first-order derivatives bounded between 0 and 1). These desirable properties often help simplify the analysis and may even be necessary to obtain an optimal contract in analytical form.
	We choose to work with $\Ac$ in \eqref{eq:Ac} to formulate Problem \ref{prob:no} for at least two reasons. First, our method does not need the extra properties on $I$ derived from the IC condition. Second, the optimal contract $I^* \in \Ac$ we obtain automatically satisfies the IC condition (see Theorem \ref{thm:1_ass}), therefore there is no need to impose it \emph{a priori}. Put differently, if we were to impose the IC condition upfront, we would obtain exactly the same result.
\end{remark}

\subsection{Optimal Contract}
\label{sub:solution}

The goal of this section is to solve Problem \ref{prob:no}, and we obtain the optimal reinsurance contract in semiclosed form in Theorem \ref{thm:1_ass}. 
We explain the key methodology of solving Problem \ref{prob:no} in a two-step procedure as follows. 
For ease of presentation, denote by $\pf$ the premium of a \emph{full} reinsurance contract, which, according to \eqref{eq:pi}, equals

\begin{align}
	\label{eq:pf}
	\pf = (1 + \eta) \Eb[X].
\end{align} 
By the definition of admissible indemnities in \eqref{eq:Ac}, the premium of an admissible contract must be between 0 (no reinsurance) and $\pf$ (full reinsurance). 

\begin{enumerate}
	\item In the first step, we fix a premium level $a \in [0, \pf]$ and only consider admissible reinsurance contracts whose premium is equal to $a$. This leads to a reduced admissible set 
	\begin{align}
		\label{eq:Ac_a}
		\Ac_a := \{ I \in \Ac \, | \,  \pi(I) = a \}, \quad a \in [0, \pf].
	\end{align}
	We solve Problem \ref{prob:no} over $\Ac_a$ and call the solution 
	\begin{align}
		\label{eq:I_op_defa}
		I^*_a = \argsup_{I \in \Ac_a} \, \Eb[u(W(I))]
	\end{align} 
    a \emph{locally} optimal reinsurance contract.

	\item In the second step, we search over all premium levels $a \in [0, \pf]$ to find the optimal premium level $a^*$, defined by 
		\begin{align}
		\label{eq:a_op}
		a^* = \argsup_{a \in [0, \pf]} \, \Eb[u(W(I_a^*))].
	\end{align}
    We call $I^* = I^*_{a^*}$ a \emph{globally} optimal reinsurance contract because it solves Problem \ref{prob:no} over $\Ac = {\cup}_{a \in [0, \pf]} \Ac_a$.
\end{enumerate}

\vspace{1ex}
\noindent
\textbf{Step 1.} We first identify a key threshold $\bar{a}$ for the premium level and show that the optimal premium $a^*$ lies in the interval $[0, \bar{a}]$. As such, we only need to solve Problem\eqref{eq:I_op_defa} over all $a \in [0, \bar{a}]$ in Step 1. To begin, we present a technical lemma for finding $\bar{a}$. 
Recall that $\pf$ is defined in \eqref{eq:pf}.

\begin{lemma}\label{lem:a_bar}
	Let $g: [0, \pf] \mapsto \Rb$ be defined by 
	\begin{align}
		\label{eq:a_bar}
		g(a) := (1 + \eta) \, \mathbb{E} \left[ X - (X - (S + a)^+)^+ \right] - a,
	\end{align}
	and define
	\begin{align}\label{eq:abars}
		\bar{a}:=\inf \left\{a\in[0,\pi_f]\mid g(a)\le 0 \right\}.
	\end{align}
    Then, $g(\bar{a})=0$ and $g(a) > 0$ for all $a \in [0,\bar{a})$.
\end{lemma}
\begin{proof}   See Appendix \ref{prooflm3.1}. 
\end{proof}

Recall that once the insurer chooses an admissible contract $I \in \Ac_a$, the total available reserve from the reinsurer is $R = (S + a)^+$ by \eqref{eq:R}. Because the contract indemnity cannot exceed the loss $X$ itself or the reinsurer's reserve $R$, it follows that for all  $I \in \Ac_a$, the actual indemnity satisfies $\Ib(I) \le X \wedge (S+a)^+ = X - (X - ( S+ a)^+)^+$. Therefore, to account for the possible counterparty default, $I(x, s) = x - (x - (s + a)^+)^+ \in \Ac_a$ is the \emph{maximum} indemnity coverage the insurer may choose among all contracts in $\Ac_a$. Note that such maximum contracts are defined for all $a \in [0, \pf]$, and they are increasing in $a$. 
However, Lemma \ref{lem:a_bar} intuitively suggests that there exists a threshold $\bar{a}$ on the premium level, and the maximum contract at level $\bar{a}$ thus serves as the absolute upper bound among all contracts in $\Ac$. To be specific, the above discussion motivates us to consider a reinsurance contract with the following indemnity function:
\begin{align}
	\bar{I}(x,s) = x - (x - ( s +\bar{a})^+)^+, \label{eq:I_bar}
\end{align}
in which $\bar{a}$ is defined in \eqref{eq:abars}.  The premium of contract $\bar{I}$ is $\pi(\bar{I}) =  \bar{a}$ because $g(\bar{a}) = 0$. 
Furthermore, as $\bar{I}(x,s) \le (s + \bar{a})^+$ for all $x \ge 0$, we have $\Ib(\bar{I}) \equiv \bar{I}(X,S)$ and $D (\bar{I}) \equiv 0$ for contract $\bar{I}$, implying that $\bar{I}$ in \eqref{eq:I_bar} is a \emph{default-free} contract (we call $I$ a \emph{default-free} contract if the reinsurer will not default when the insurer chooses contract $I$).  
The discussion so far suggests that contract $\bar{I}$ in \eqref{eq:I_bar} serves as a ``threshold'' on the admissible indemnities, as confirmed below. 

\begin{proposition}\label{prop:a_bar}
	For all $a \in [\bar{a}, \pf]$ and all $I \in \Ac_a$, we have  
	\begin{align}
    \label{eq:ineq_Ibar}
		\Eb \big[ u(W(\bar{I}))\big] \ge \Eb \big[ u(W(I))\big],
	\end{align} 
    in which  $\bar{I} \in \Ac_{\bar{a}}$ is defined by \eqref{eq:I_bar}. In addition,  if $S\ge 0$ almost surely (i.e., $\Pb(S\ge 0)=1$) and $\eta=0$, then $ \bar{I}=I^*$ is the globally optimal reinsurance contract to Problem \ref{prob:no}.
\end{proposition}
\begin{proof}   See  Appendix \ref{prop3.1} . 
\end{proof}

Proposition \ref{prop:a_bar} offers two important insights. First, consider the case of $S\ge 0$  and $\eta = 0$; the proposition shows that the optimal reinsurance is $\bar{I}$, a contract with \emph{partial} coverage. If the reinsurer's default is otherwise ignored, \citet{mossin1968aspects} shows that the optimal contract when $\eta = 0$ is full coverage ($I^*(x) = x$). As such, Proposition \ref{prop:a_bar} extends Mossin's result by incorporating the counterparty default risk. Note that the default risk does not impact the deductible choice, as $\bar{I}$ has a zero deductible (noting $\bar{I}(x,s) = x$ for all $x < (s + \bar{a})^+$), the same as in \citet{mossin1968aspects}. However, because of the possible default from the reinsurer, the insurer will not seek full coverage even when the loading $\eta$ is zero. Instead, the optimal contract has a policy limit of $(S + \bar{a})^+$, which equals the reinsurer's reserve.
Second, as $\bar{I}$ dominates all admissible $I$ with premiums greater than $\bar{a}$, the optimal premium level $a^*$ defined in \eqref{eq:a_op} is achieved in $[0, \bar{a}]$. As such, the remaining task in Step 1 is to solve \eqref{eq:I_op_defa} for all $a \in [0, \bar{a}]$, and the next theorem completes this task.

\begin{theorem}
	\label{thm:step_one}
     For every $a \in [0, \bar{a}]$,  the locally optimal reinsurance contract $I^*_a$ to Problem \ref{prob:no} over the constrained set $\Ac_a$ in \eqref{eq:Ac_a}  is given by 
	\begin{align}
		\label{eq:I_op}
		I_a^*(x,s) = (x - d(a))^+ - (x - d(a) - (s + a)^+)^+,
	\end{align}
    in which the deductible amount \( d(a) \in [0, M] \) solves the equation 
    \begin{align}
    	\label{eq:g_a}
    	g_a(y):= (1+\eta)\Eb[(X- y)^+-(X-y-(S+a)^+)^+] - a = 0. 
    \end{align}
\end{theorem}
\begin{proof}   See Appendix \ref{thm3.1}. 
\end{proof}

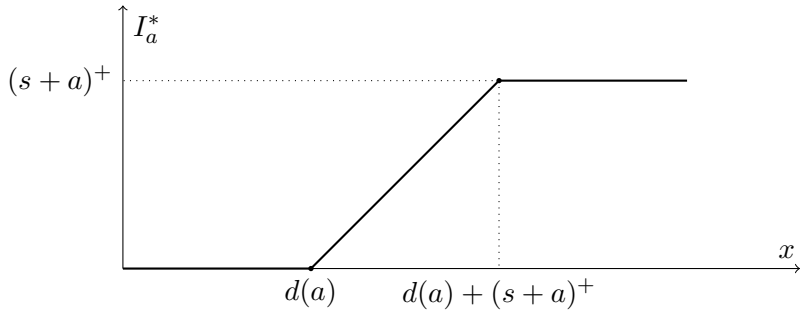
\begin{figure}[h]
	\centering
	\vspace{-2ex}
	\begin{tikzpicture}[scale = 0.5]
		\draw[->] (0,0) -- (18,0) node[above,xshift=-5pt] {$x$};
		\draw[->] (0,0) -- (0,7) node[below right] {$I_a^*$};
		
		\foreach \x in {5,10} 
		\draw (\x,0) node[below] {\ifnum\x=5 $d(a)$ \else $d(a) + (s+a)^+$ \fi};
		\foreach \y in {5} 
		\draw (0,\y) node[left] { $(s+a)^+$ };
		
		\draw[domain=0:5, thick] plot (\x,0);
		\draw[domain=5:10, thick] plot (\x, \x-5);
		\draw[domain=10:15, thick] plot (\x, 5);
		
		\draw[dotted] (10,5) -- (10,0);
		\draw[dotted] (10,5) -- (0,5);
		
		\fill (5, 0) circle (2pt);
		\fill (10, 5) circle (2pt);
		
	\end{tikzpicture}
	\vspace{-1ex}
	\caption{Optimal contract $I^*_a$ in \eqref{eq:I_op}.}
	\label{fig:I_op}
    \vspace{-2ex}
\end{figure}

Several remarks on Theorem \ref{thm:step_one} are due as follows. 
The optimal contract $I^*_a$ is obtained in semiclosed form in \eqref{eq:I_op}, and the only unknown $d(a)$ can be easily computed by a numerical method. 
As easily seen from Figure \ref{fig:I_op}, $I^*_a$ is a deductible reinsurance contract with a  policy limit; the deductible amount is $d(a)$, and the maximum covered loss is $d(a) + (s + a)^+$, yielding a cap of $(s + a)^+$ on the contractual indemnity. Such a contract structure is economically justifiable and commonly used in practice. 
On the one hand, the deductible is present due to the cost of risk transferring and indeed vanishes when $\eta = 0$ (i.e., we have $d(a) = 0$ in \eqref{eq:I_op} when $\eta = 0$, and this result is proved in Proposition \ref{prop:a_bar}).
On the other hand, the policy limit is the direct consequence of the reinsurer's default risk.
We observe from Figure \ref{fig:I_op} that $I^*_a(x, s) \le  (s + a)^+$ for all $x \ge 0$, and thus $I^*_a$ is a default-free  contract (i.e., $D(I^*_a) \equiv 0$ in \eqref{eq:D}).

Indeed, a contract that depends on the reinsurer's background risk \(S\) is efficient because it allows the insurer to tailor the coverage to the reinsurer's ability to pay. Default occurs when the promised indemnity \(I(X, S)\) exceeds the reinsurer's available resources \((S + \pi(I))^+\). If the insurer chooses a contract that may extend the reinsurer's reserve, this would create a deadweight cost in default states, in the sense that the insurer pays for a coverage $I(X,S)$ but receives strictly smaller $\tau (S + \pi(I))^+$. As such, the optimal contract $I^*_a$ should always be within the reinsurer's reserve (i.e., it is default-free); otherwise, one can easily construct a default-free contract that strictly dominates $I^*_a$, contracting its optimality (the exact proof is available upon request).

Thanks to Theorem \ref{thm:step_one}, we can reduce Problem \ref{prob:no}, an infinite-dimensional optimization problem, into a one-dimensional scalar optimization problem in \eqref{eq:a_op}, which we study in the next step.  The solution $d(a)$ to \eqref{eq:g_a} may not be unique in general. However, with mild assumptions imposed on the insurer's loss $X$ (see Corollary \ref{cor:unique_d}), the uniqueness of $d(a)$ is gained.

\vspace{1ex}
\noindent
\textbf{Step 2.}  We solve \eqref{eq:a_op} to obtain the optimal premium level $a^*$ and identify the globally optimal reinsurance contract  to Problem \ref{prob:no} as $I^* = I^*_{a^*}$. 
To that end, we make rather mild assumptions on the insurer's loss $X$ and the reinsurer's background risk $S$, and discuss them in Remark \ref{rem:X_dist}.

\begin{assumption}\label{ass} The insurer's loss $X$ and the reinsurer's background risk $S$ satisfy the following conditions: (1) $S \ge  0$ almost surely (i.e., $\Pb(S\ge 0)=1$); (2) both $X$ and $X-S$ have a finite number of jump points on $[0,M]$; (3)
$\Pb(X\le x)$ strictly increases with respect to  $x \in[0,M]$; (4) $\Eb[(X-y)^+-(X-y-S)^+]>0$ holds for all $y\in[0,M)$. 
\end{assumption}

\begin{remark}
	\label{rem:X_dist}
	The reinsurer is required to set aside a strictly positive initial reserve at time 0 (the inception time of a contract), then Condition (1) in Assumption \ref{ass}, also imposed in \citet{biffis2012optimal}, simply means that risky investments in the reinsurer's reserve, such as equities, have limited liabilities, consistent with most real-life scenarios.
	By Condition (2), both $X$ and $X-S$ can have jumps at any point, but the total number of jumps on $[0,M]$ must be finite.  In the literature, similar, but stronger, assumptions are often imposed in the study of optimal (re)insurance problems. For instance, \citet{bernard2015optimal} assume that $X$ has no atom, while \citet{cai2014optimal} assume that $X$ only has a jump at $0$. Condition (3) is also imposed in \citet{asimit2013optimal}  and  \citet{cai2014optimal}. Condition (4) is a rather mild condition and holds in real markets, because Condition (3), along with \(S > 0\), implies Condition (4). 
	
	On the technical side, because of $S \ge 0$, we have \( (S + a)^+ = (S + a) \) for all $a \ge 0$, and it helps avoid discontinuities when taking derivatives with respect to \( a \). Condition (2) is used to show that certain functions are continuously differentiable, except for a finite set. Conditions (3) and (4) ensure the uniqueness of some solutions. 
	These properties are used in the proofs of Corollary \ref{cor:unique_d}, Theorem \ref{thm:1_ass}, and Proposition \ref{prop:a_p}.

	In the special case \emph{without} the reinsurer's background risk, we have $S \equiv r > 0$, the initial reserve; recall that $S \le 0$ is already analyzed in Proposition \ref{prop:extre}. In this case, we can further remove the conditions on $S$ in Assumption \ref{ass}.  
\end{remark}

Recall that for a fixed $a \in [0, \bar{a}]$, the deductible $d(a)$ that solves $g_a(y) = 0$ in Theorem \ref{thm:step_one} may not be unique. However, under Assumption \ref{ass},  the next corollary shows that $d(a)$ is unique. 
\begin{corollary}\label{cor:unique_d}
	Let Assumption \ref{ass} hold. For every $a \in [0, \bar{a}]$, there exists a unique solution $d(a) \in [0, M]$ to  $g_a (y) = 0$ in \eqref{eq:g_a}.
\end{corollary}
\begin{proof}   See  Appendix \ref{Coroly3.1}.  
\end{proof}

We now solve for the optimal premium level $a^*$ and obtain a full solution to the insurer's optimal reinsurance problem (see Problem \ref{prob:no}). Note that when $S \ge 0$  and $\eta = 0$, Proposition \ref{prop:a_bar} shows that $\bar{I}$ in \eqref{eq:I_bar} is the optimal contract to Problem \ref{prob:no}.  Also, recall that $\bar{a}$ is defined by \eqref{eq:abars}.

\begin{theorem}\label{thm:1_ass}
	Let Assumption \ref{ass} hold. The globally optimal reinsurance contract $I^*$ to Problem \ref{prob:no} is given by 
	\begin{align}
		\label{eq:I_star}
		I^*(x,s) = \begin{cases}
			x-(x-(s+\bar{a}))^+, & \text{if } \eta = 0, \\
			(x-d(a^*))^+-(x-(d(a^*)+s+a^*))^+, & \text{if } 0<\eta<\dfrac{u'(w-M)}{\Eb[u'(w-X)]}-1, \\
			0, & \text{if } \eta\ge \dfrac{u'(w-M)}{\Eb[u'(w-X)]}-1,
		\end{cases}
	\end{align}
in which $M = \mathrm{ess \, sup} \, X < \infty$, $d(a)$ is established in Corollary \ref{cor:unique_d} for all $a \in [0, \bar{a}]$, and $a^*\in (0,\bar{a})$ is the unique solution to 
\begin{align}
	\label{eq:a_star}
	\Eb[u'(w-(X\wedge d(a))-a)]-\frac{u'(w-d(a)-a)}{1+\eta}=0.
\end{align}
\end{theorem}

\begin{proof}   See  Appendix \ref{proof:1_ass}.  
\end{proof}

 Theorem \ref{thm:1_ass} presents the optimal contract $I^*$ case by case based on the value of the premium loading $\eta$. Alternatively, we can write $I^*$ in the following uniform  expression: 
 \begin{align}
 	I^*(x,s) = (x-d(a^*))^+-(x-(d(a^*)+s+a^*))^+,
 \end{align}
because 
 \begin{align}
 	\label{eq:a_star_ex}
 	\begin{cases}
 		a^* = \bar{a} \text{ and } d(\bar{a}) = 0, & \text{ if } \eta = 0, \\[1ex]
 		a^* = 0 \text{ and } d(0) = M, & \text{ if } \eta\ge \frac{u'(w-M)}{\Eb[u'(w-X)]}-1.
 	\end{cases}
 \end{align}
As suggested by the uniform expression, the optimal contract $I^*$ is a deductible reinsurance with a policy limit, just as $I^*_a$ in \eqref{eq:I_op} (see Figure \ref{fig:I_op}). We remark that the presence of the policy limit in $I^*$ reflects the impact of counterparty default on the insurer's decision making; the policy limit vanishes when the reinsurer's reserve is sufficiently large (so that default never occurs). 
The first case in \eqref{eq:I_star} shows that the deductible in $I^*$ vanishes when the premium loading $\eta$ equals zero.
From the last case in \eqref{eq:I_star}, we see that if the premium loading $\eta$ is too high, the insurer is better off with 100\% self-insurance than purchasing reinsurance from the reinsurer. 
If reinsurance is reasonably priced as in the second case of \eqref{eq:I_star}, endogenous default may occur if the insurer chooses an arbitrary contract among all admissible choices in \eqref{eq:Ac}. 
However, if the insurer chooses the optimal contract $I^*$, we always have $D(I^*) \equiv 0$, and the reinsurer will never default on $I^*$. On a technical note, for the second case, we need to solve a nonlinear equation \eqref{eq:a_star} to get the optimal premium $a^*$, which may seem to be a challenging task. Luckily, we can show that both $g_a$ in \eqref{eq:g_a} and the left-hand side of \eqref{eq:a_star} are strictly decreasing functions, which allows us to efficiently compute $a^*$ and $d(a^*)$ (see Example \ref{exa2} below). 
Last, we observe that the optimal contract $I^*$ satisfies the IC constraint automatically (i.e., $0 \le I^*(x,s) - I^*(x',s) \le x - x'$ for all $0 \le x' \le x \le M$), which is why we do not impose the IC constraint upfront in defining the admissible set $\Ac$ in \eqref{eq:Ac}.

Due to the presence of endogenous default, we cannot establish the \emph{strict} concavity of the functional \( \mathcal{J}(I):= \mathbb{E}[u(W(I))] \). Fortunately, 
we can show that under Assumption \ref{ass}, \( I^* \) in \eqref{eq:I_star} is the unique globally optimal reinsurance contract.
\begin{proposition}\label{prop:unique}
	Let Assumption \ref{ass} hold. Then, $I^*$ in \eqref{eq:I_star} is the unique globally optimal reinsurance contract to Problem \ref{prob:no}.
\end{proposition}
\begin{proof} 
  See Appendix \ref{prop3.2}.  
\end{proof} 

\subsection{Comparative Statics}
\label{sub:comp}

In this section, we conduct a comparative statics analysis on the optimal reinsurance contract $I^*$ obtained in Theorem \ref{thm:1_ass}. This goal can be easily achieved by a \emph{numerical} method once the model is given, but it is challenging to obtain \emph{analytical} results, which we are able to achieve under mild conditions (see Proposition \ref{prop:a_p}). Note that certain, but \emph{not} all, results in Proposition \ref{prop:a_p} requires a condition on the insurer's utility function, as stated below.

\begin{assumption}
	\label{asu:u}
	The insurer's utility function $u$ satisfies the decreasing absolute risk aversion (DARA) condition; i.e., the Arrow-Pratt coefficient of absolute risk aversion, defined by $\mathbb{A}_u= - \frac{u''}{u'}$, is a decreasing function.
\end{assumption}

By definition, agents with DARA risk preferences have reduced risk aversion when their wealth increases. This result is mostly consistent with empirical findings (see \citet{levy1994absolute}).
A prominent example of DARA risk preferences is the family of power utility functions $u(x) = \frac{1}{1 - \gamma} \, x^{1 - \gamma}$, in which $\gamma > 0$ and $\gamma \neq 1$.

Before we present the key results on comparative statics, we introduce the following notations: 
let $a^*$ denote the optimal premium level for all $\eta \ge 0$ (as defined in \eqref{eq:a_op} and calculated by \eqref{eq:a_star} or \eqref{eq:a_star_ex}); for the optimal reinsurance contract, $d^*:= d(a^*)$ is the deductible amount, and $U^* := a^* + d^* + S$ is the policy limit (maximum covered loss). The following proposition summarizes the \emph{analytical} results on how the specifications of the insurer's optimal contract ($a^*$, $d^*$, and $U^*$) are affected by model inputs. Because we allow discontinuities in the distribution of the insurer's loss $X$, the proof is technical and lengthy, and thus we defer it to Online Companion.

\begin{proposition}\label{prop:a_p}
	Let Assumption \ref{ass} hold. We have the following comparative statics results on the optimal reinsurance contract $I^*$ in \eqref{eq:I_star}. 
	\begin{enumerate}
		\item The optimal premium level $a^*$ increases with respect to the reinsurer's background risk $S$ (in the pointwise sense),
		so is $a^* + d^*$. Furthermore, if Assumption \ref{asu:u} holds,  then the optimal deductible $d^*$ decreases with respect to $S$. 
	
		\item If Assumption \ref{asu:u} holds, then the optimal premium level $a^*$ decreases with respect to the insurer's initial wealth $w$, but both the optimal deductible $d^*$ and $a^* + d^*$ increase with respect to $w$.

		\item The optimal premium level $a^*$ increases with respect to the insurer's (Arrow-Pratt) degree of risk aversion $\mathbb{A}_u$,
		but both the optimal deductible $d^*$ and $a^* + d^*$ decrease with respect to $\mathbb{A}_u$.
	\end{enumerate}
\end{proposition}

\begin{proof} 
		See Online Companion II.  
\end{proof}

In the first item of Proposition \ref{prop:a_p}, we analyze the impact of the reinsurer's background risk $S$ on the optimal contract; note that ``increase'' or ``decrease'' is in the pointwise sense (i.e., if $S_1$ increases to $S_2$, we have $S_2(\omega) \ge S_1(\omega)$ for all $\omega \in \Omega$, except for a negligible set). Recall that $S + \pi(I)$ is the total available reserve for settling claims at time 1, thus the larger the $S$, the lower the default possibility. 
As such, when $S$ increases, we expect the insurer to seek more coverage for its risk exposure, which is confirmed by the decrease of the deductible $d^*$ and the increase of the policy limit $U^*$ (recall $U^* = a^* + d^* + S$); the increase of coverage naturally means a higher premium paid to the reinsurer.

Next, we study how the insurer's initial wealth $w$ affects the optimal contract. As implied by the very definition of DARA risk preferences, when $w$ increases, the insurer's risk aversion decreases, and thus its demand for reinsurance coverage reduces, leading to a higher deductible and a lower premium, both of which are consistent with the results in \citet{mossin1968aspects} and \citet{schlesinger1981optimal}. 
In addition, the impact of $w$ on $d^*$ is more significant than that on $a^*$, which is why the policy limit $U^* = a^* + d^* + S$ changes in the same direction as $d^*$ when $w$ changes.

Last, to investigate the impact of risk aversion on reinsurance decision, consider two insurers with different utility functions, $u_1$ and $u_2$, and call them Insurer 1 and Insurer 2, respectively.
We say that Insurer 1 has a higher (Arrow-Pratt) degree of risk aversion than Insurer 2 if \[ \mathbb{A}_{u_1}(x) = - \frac{u_1''(x)}{u_1'(x)} \ge \mathbb{A}_{u_2}(x) = - \frac{u_2''(x)}{u_2'(x)},\] 
and assume, without loss of generality, that this is the case.
Because  risk aversion is a key driver for reinsurance, we anticipate that Insurer 1 chooses a contract with a lower deductible and is willing to spend more on reinsurance than Insurer 2 does, as confirmed by Item 3 of Proposition \ref{prop:a_p}.

We close this section with a numerical example, and it serves two purposes. First, we demonstrate that finding the optimal contract $I^*$ is an easy task once the model parameters are given. Second, we use the example to visualize the analytical results obtained in Proposition \ref{prop:a_p}.

\begin{example}\label{exa2}
	Consider an insurer with a power utility function $u(x) = \frac{1}{1 - \gamma} \, x^{1 - \gamma}$, in which $\gamma > 0$ and $\gamma \neq 1$. Assume that the insurer's loss distribution has a probability mass of 10\% at 0 and 10\% at 10 (i.e., $\Pb(X = 0) = \Pb(X=10) = 10\%$) and a continuous density function over $(0,10)$, which takes the form of a truncated Pareto density, 
	\begin{align}
		f_X(x) = \frac{96}{35}\frac{10^3}{(x+10)^4} \,\mathbf{1}_{\{x\in(0,10)\}} .
	\end{align}
    We also assume that there is no background risk, then $S$ equals the reinsurer's initial reserve $r >0$.
\end{example}

First, we fix $\gamma = 1/2$, $w=15$, and $S=r=5$.  The threshold value of the premium loading $\eta$ in \eqref{eq:I_star} is calculated by $\frac{u'(w-M)}{\Eb[u'(w-X)]} - 1 = 0.4669$. We plot the optimal premium $a^*$ and the optimal deductible $d^*=d(a^*)$ as a function of $\eta$, the premium loading, over $[0, 0.5]$ in Figure \ref{fig:eta}. When $\eta \ge  0.4669$, Figure \ref{fig:eta} shows that $d^* = M$, or equivalently $I^* \equiv 0$, which confirms Case 3 in \eqref{eq:I_star}. When $0 < \eta < 0.4669$, we numerically solve $a^*$ from \eqref{eq:a_star}, and the left panel of Figure \ref{fig:eta} shows that $a^*$ is a strictly decreasing function of $\eta$ in this range. Because reinsurance contracts become more expensive when $\eta$ increases, the insurer reacts to the price increase by choosing a higher deductible, which is confirmed by the right panel of Figure \ref{fig:eta}.

\begin{figure}[htbp]
	\centering
	\vspace{-2ex}
	\begin{minipage}[b]{0.49\linewidth}
		\centering
		\includegraphics[width=\linewidth, height = 5cm]{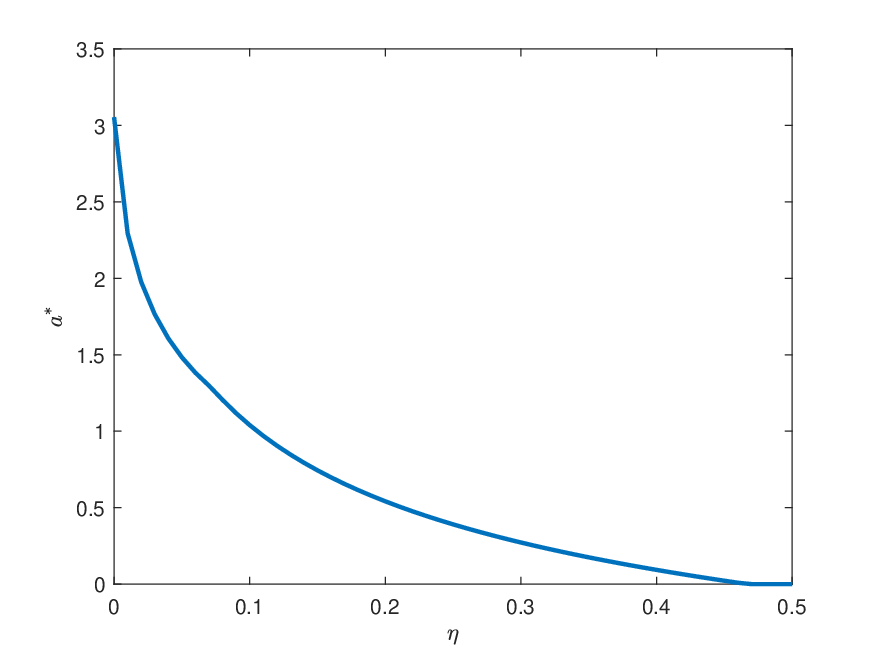}
	\end{minipage}
	\begin{minipage}[b]{0.49\linewidth}
		\centering
		\includegraphics[width=\linewidth, height = 5cm]{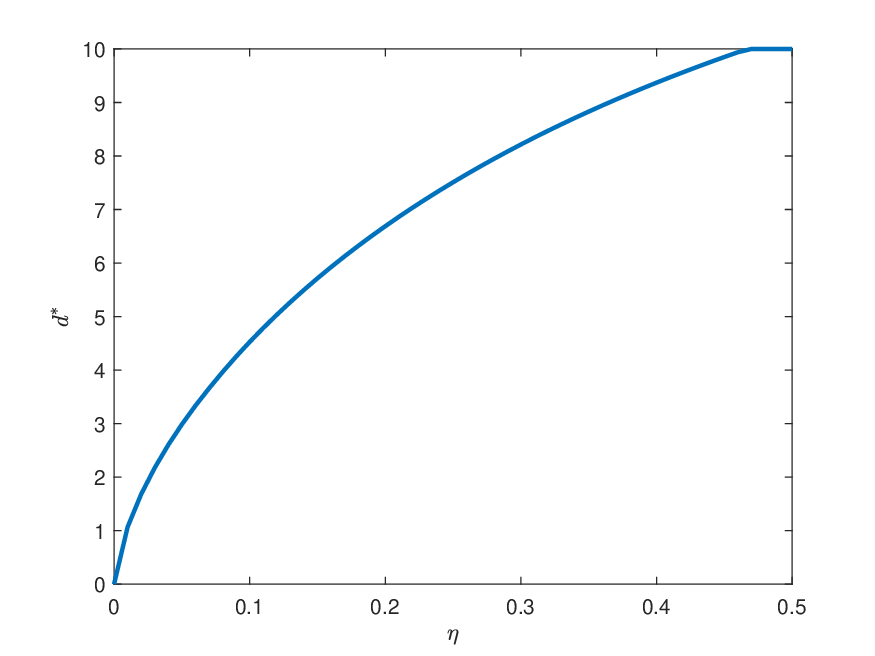}
	\end{minipage}
	\vspace{-1ex}
	\caption{Optimal premium $a^*$ (left) and deductible $d^*$ (right) with respect to  $\eta$}
	\label{fig:eta}
    \vspace{-2ex}
\end{figure}
Next, we investigate the impact of $\gamma$, the insurer's risk aversion, on the optimal contract. In this study, we fix the premium loading $\eta = 0.2$ and the insurer's initial wealth $w = 15$, but consider two different levels for the reinsurer's reserve $r$, $r=2$ (low reserve case) and $r=8$ (high reserve case). With those parameters, we plot the optimal premium $a^*$ and the optimal deductible $d^*=d(a^*)$ as a function of $\gamma$ over $(0.2, 3)$ ($1-\gamma$ over $(-2,0.8)$) in Figure \ref{fig:gamma}. For both the low and high reserve cases, we observe the same sensitivity results: when the risk aversion $\gamma$ increases,the optimal premium level $a^*$ increases, but the optimal deductible $d^*$ decreases, both of which confirm the results in Item 3 of Proposition \ref{prop:a_p}. 
As shown in the right panel of Figure \ref{fig:gamma}, the two optimal deductible $d^*$ curves are really close because $d^*$ is somehow insensitive to $r$, but we still observe that $d^*|_{r = 2} > d^*|_{r = 8}$, which confirms Item 1 of Proposition \ref{prop:a_p}.
The left panel of Figure \ref{fig:gamma} demonstrates the significant impact of endogenous default on the optimal contract: in the high reserve case of $r=8$, the reinsurer essentially never defaults, thus $a^*$ under $r=8$ (dotted line) coincides with the optimal premium level in a model without default risk. However, in the low reserve case of $r = 2$, the reinsurer's default risk becomes a major concern to the insurer; in response, the insurer chooses a much lower policy limit, resulting in a sharp decrease in the optimal premium $a^*$, compared to the high reserve case.

\begin{figure}[htbp]
	\centering
	\vspace{-2ex}
	\begin{minipage}[b]{0.49\linewidth}
		\centering
		\includegraphics[width=\linewidth, height = 5cm]{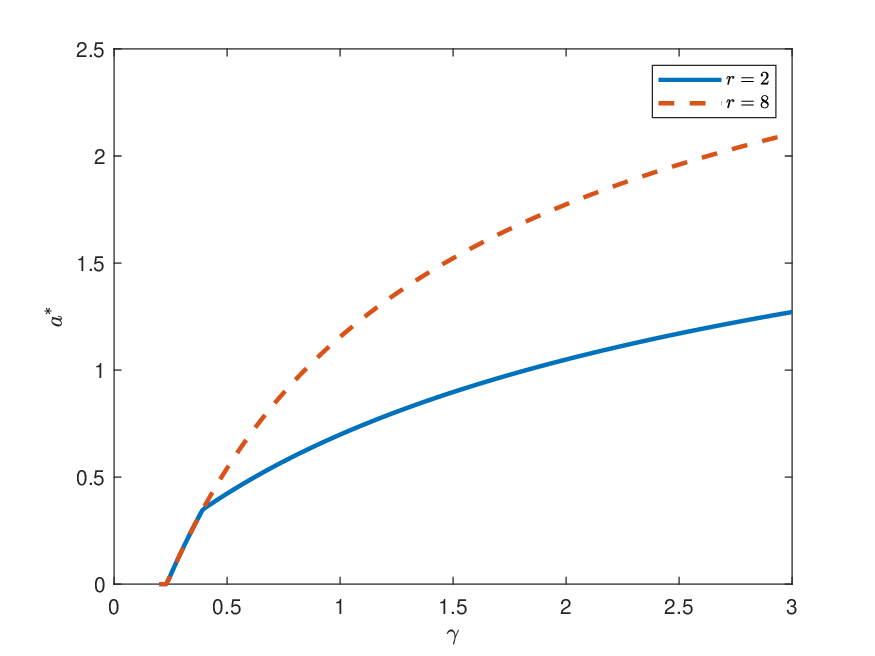}
	\end{minipage}
	\begin{minipage}[b]{0.49\linewidth}
		\centering
		\includegraphics[width=\linewidth, height = 5cm]{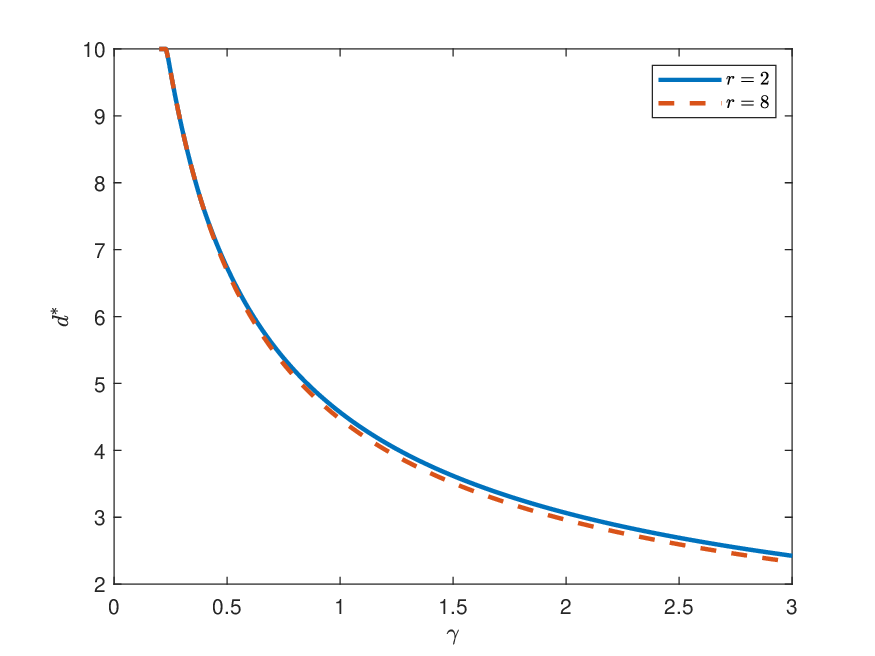}
	\end{minipage}
	\vspace{-1ex}
	\caption{Optimal premium $a^*$ (left) and deductible $d^*$ (right) with respect to $\gamma$}
	\label{fig:gamma}
    \vspace{-2ex}
\end{figure}

\section{Optimal Loss-Dependent Indemnities}\label{sec:ext}

In this section, we consider loss-dependent indemnities in the form of $I = I(x)$; namely,  the insurer is only allowed to choose contracts that depend on the loss size $x$ itself, but \emph{not} on  the reinsurer's background risk level $s$. The method that we use to solve Problem \ref{prob:no} in the previous section no longer applies here;
in particular, we do not have a proper ``upper'' bound on the default-free indemnities for the problem in this section as $\bar{I}$ in \eqref{eq:I_bar} for Problem \ref{prob:no}.
We will develop a different approach to obtain the optimal loss-dependent reinsurance contract.

Recall from Problem \ref{prob:no} that admissible indemnities there only need to satisfy the ``principle of indemnity'' condition, and it yields the most general (largest) set $\Ac$ in \eqref{eq:Ac} (see Remark \ref{rem:ind}).
But in this section, we further impose the IC condition on admissible indemnities,
and this leads to the following admissible set $\widetilde{\Ac} = \AcIC$:
\begin{align}
	\label{eq:Acl}
	\AcIC :=\{I: [0,M] \mapsto \Rb_+ \, | \, I(0) = 0 \text{ and } 0 \le I(x) - I(x') \le x - x' \, \text{for all}\, x \ge x' \ge 0 \}.
\end{align}
Note that the IC condition helps rule out certain \emph{ex post} moral hazard (see, e.g., \citet{huberman1983optimal}, \citet{xu2019optimal}, and \citet{boonen2022mean}) and is also imposed in several related works (see \citet{asimit2013optimal} and \citet{cai2014optimal}). 
In the proofs, we use the IC condition to handle contractual indemnity in a piecewise way and to get the monotonicity of \( x \) and \( x - I(x) \), as well as to establish an up-crossing property which we use to prove Theorem \ref{thm:N}.

For every $I \in \AcIC$, denote its actual indemnity by $\Ib_S$. 
By \eqref{eq:Ib}, we have 
\begin{align}
	\Ib_S(I)
	&=  I(X) \cdot \mathbf{1}_{\{I(X) \leq (S + \pi(I))^+\}} + \tau (S + \pi(I))^+ \cdot \mathbf{1}_{\{I(X) > ( S + \pi(I))^+\}}, 
	\label{eq:Ib_S}
\end{align}
in which $\tau \in [0,1]$ is the recovery rate.
As such, for a chosen contract $I \in \AcIC$, the insurer's terminal wealth, $W_S(I)$, equals 
\begin{align}\label{eq:Ws}
	W_S(I) = w - X + \Ib_S(I) - \pi(I).
\end{align}

We now formulate the second concrete version of Problem \ref{prob:main} as follows.

\begin{problem}	\label{prob:s}
	The insurer seeks an optimal loss-dependent reinsurance contract $I_S^*:= I_S^*(x) \in \AcIC$  to maximize the expected utility of its terminal wealth under the reinsurer's endogenous default risk and background risk, i.e., $$I_S^* = \argsup_{I \in \AcIC} \, \Eb[u(W_S(I))],$$
in which the admissible set $\AcIC$ is defined in \eqref{eq:Acl}, and $W_S(I)$ is given by \eqref{eq:Ws}.  
\end{problem}

Even with the additional IC condition imposed in \eqref{eq:Acl}, solving Problem \ref{prob:s} for a general $S$ is still unlikely. In comparison, Theorem \ref{thm:step_one} solves Problem \ref{prob:no} (with $I = I(x, s)$) over $\Ac_a$ without imposing any assumption on the distribution of $S$, and Theorem \ref{thm:1_ass} finds the global solution to Problem \ref{prob:no}  over $\Ac$ with mild assumptions (Assumption \ref{ass}). In order to obtain an analytical solution to Problem \ref{prob:s}, we impose the following assumptions on $S$.

\begin{assumption}\label{ass2}
	The reinsurer's background risk \(S\) is independent of the insurer's loss \(X\) and follows an \(N\)-point discrete distribution 
	\begin{align}
		\Pb(S = s_i) = p_i >0, \quad i = 1, \cdots, N,
	\end{align}
	in which $N$ is an arbitrary positive integer, $s_1 < s_2 < \cdots < s_N$ and $\sum_{i=1}^N p_i = 1$. Let the set of possible values of \( S \) be denoted by \( \Sc := \{s_1, s_2, \cdots, s_N\} \).
\end{assumption}

Recall that $S$ can be interpreted as random shocks (e.g., from the financial markets) to the reinsurer's reserve,  thus the independence assumption is overall reasonable (see Section III in \cite{doherty1983optimal} for the same assumption and Section 7.4.2 in \cite{schlesinger2013theory} for further discussion and related reference). The assumption that $S$ is a discrete random variable allows us to solve for the optimal contract conditioning on each realization $S = s_i$ (see Remark \ref{rem:proof}), and this assumption is not restrictive because $N$ can be arbitrarily large and there are no additional conditions on the realized values $s_1, \ldots, s_N$ or their  probabilities $p_1, \ldots, p_N$.

Under Assumption \ref{ass2}, we rewrite the insurer's objective as 
 \begin{align*}
 	\mathbb{E}[u(W_S(I))] 
 	&= \sum_{i=1}^N  \, p_i \, \mathbb{E} \left[u \left(w - X + \Ib_{s_i}(I) - \pi(I) \right) \right],
 \end{align*}
where $\Ib_{s_i}$ is defined by \eqref{eq:Ib_S} with $S = s_i \in \Sc$ for all $i=1, \cdots, N$. An application of Proposition \ref{prop:extre} directly yields the following result, and thus we omit its proof.

\begin{corollary}
	\label{cor:negative}
	Let Assumption \ref{ass2} hold and recall $s_N = \max \, S$. If $s_N \le 0$, then the optimal contract to Problem \ref{prob:s} is no reinsurance, $I^*_S \equiv 0$. 
\end{corollary}
Thanks to Corollary \ref{cor:negative}, we focus on Problem \ref{prob:s} when $s_N > 0$ in the rest of the section.  The next lemma identifies a upper bound on the optimal premium.

\begin{lemma}\label{lem:a_barN}
	Suppose $s_N>0$. There exists a unique solution $\bar{a}_N$ over $[0, \pf]$ to the equation 
	\begin{align}
		\label{eq:a_barN}
		g_N(a):= (1 + \eta) \, \mathbb{E} \left[ X - (X - (s_N + a))^+ \right] - a = 0.
	\end{align}
	Moreover, $g_N(a) > 0$ for all $a \in [0, \bar{a}_N)$, and $g_N(a) < 0$ for all $a \in (\bar{a}_N, \pf]$.
\end{lemma}
\begin{proof}
	See Appendix \ref{prooflm4.2}.
\end{proof}

\begin{theorem}\label{thm:N}
	Let Assumption \ref{ass2} hold. The optimal reinsurance contract \(I_S^*\) to Problem \ref{prob:s} over \(\AcIC\) is of the following parametric form:
	\begin{align}\label{eq:opI}
		I_S^*(x) = \sum_{i=1}^{N} \left[ \big(x - l_i - (  \pi(I^*_S) + s_{i-1})^+ \big)^+ - \big(x - l_i - (  \pi(I^*_S) + s_i)^+ \big)^+\right],
	\end{align}
	in which \(s_0 := -r - \pi(I^*_S)\), and the constants \(\{l_i\}_{i=1, \cdots, N}\) are free parameters satisfying the constraints \(0 \leq l_i \leq l_{i+1} \leq M\). Moreover, if $s_N > 0$, then $\pi(I_S^*)\in[0,\bar{a}_N]$.
\end{theorem}
\begin{proof} 
  See Appendix \ref{thm4.1}.   
\end{proof} 

\begin{remark}
	\label{rem:proof}
	In this technical remark, we explain the essential idea that helps us obtain Theorem \ref{thm:N}. For every $I \in \AcIC$, there exists a critical point $x_i$ corresponding to each realization of the background risk $S = s_i$, $i=1,\cdots,N$, so that the reinsurer defaults if the loss $X$ exceeds this critical value $x_i$, given $S= s_i$. These $N$ critical points, along with $x_0:=0$ and $x_{N+1}:= M$, partition the loss domain $[0, M]$ into $N + 1$ sub-intervals, $A_1 = [0,x_1]$ and $A_i = (x_{i-1}, x_i]$, $i=2, \cdots, N+1$. We prove Theorem \ref{thm:N} in the following two steps:
	\begin{enumerate}
		\item For every $I \in \AcIC$, we construct a new indemnity $I_1$ satisfying two key conditions:\\
		(i)  $\Eb \left[I_1(X) \, \mathbf{1}_{X \in A_i} \right] = \Eb \left[I(X) \, \mathbf{1}_{X \in A_i} \right]$, for all $i = 1, \cdots, N+1$, and each equation determines one parameter in the proposed form of $I_1$; combining all $N+1$ equations yields $\Eb \left[I(X)\right] = \Eb \left[I_1(X)  \right]$.
		\\ (ii) $I_1$ ``up-crosses'' $I$ in each $A_i$, so that $\Eb [u(W_S(I_1)) \, \mathbf{1}_{X \in A_i}] \ge \Eb [u(W_S(I)) \, \mathbf{1}_{X \in A_i}]$. 
		\\
		As such,  we have $\pi(I_1) = \pi(I)$ and $\Eb [u(W_S(I_1)) ] \ge  \Eb [u(W_S(I))]$. 
		
		\item For any $x \in A_{N+1}$, the indemnity $I_1(x)$ necessarily results in default, which is intuitively suboptimal. Therefore, we construct an alternative indemnity $I_2$ such that $A_{N+1} = \emptyset$. Specifically, based on $I_1$ from Step 1, we construct an indemnity $I_2$, which takes the form in \eqref{eq:opI} and involves $N$ parameters (recall that $I_1$ contains $N+1$ parameters). We show that $\Ib_S(I_1) - \Ib_S(I_2) \le \pi(I_1) - \pi(I_2)$ and, by recalling the definition of $W_S$ in \eqref{eq:Ws}, $W_S(I_2) \ge W_S(I_1)$, verifying the optimality of $I_S^*$ in \eqref{eq:opI} to Problem \ref{prob:s}.   
	\end{enumerate}
\end{remark}

\begin{remark}
In this remark, we consider the case where \(S\) is replaced by \(\mathrm{VaR}_{\alpha}(I(X))\), which depends on the indemnity \(I\), for a given threshold \(\alpha \in (0,1)\). Using a similar constructive approach, we obtain the optimal reinsurance contract $I^*$ over $\mathcal{A}_{\textup{IC}}$ in the following parametric form:
    \begin{align*}
        I^*(x)=(x-d_1)^+-(x-\mathrm{VaR}_{\alpha}(X))^++(x-d_2)^+-(x-(d_2+\pi(I^*)))^+,
    \end{align*}
    in which the two parameters $d_1$ and $d_2$ satisfy $0\le d_1\le \mathrm{VaR}_{\alpha}(X)\le d_2$. The detailed proof is provided in Online Companion. This result is consistent with \citet{cai2014optimal}, but our constructive approach is more concise and requires weaker assumptions. Furthermore, by varying premium levels, we show that the optimal contract \(I^*\) above is default-free and features a policy limit of \(I^*(\mathrm{VaR}_{\alpha}(X)) + \pi(I^*)\), a finding that aligns with Theorem \ref{thm:step_one}.
\end{remark}

Although both Problems \ref{prob:no} and \ref{prob:s} take into account the reinsurer's background risk, they are largely different from the mathematical viewpoint, as seen from Remark \ref{rem:proof} and the proofs (in the appendix). Upon examining their (local) solutions $I_a^*(x,s)$ in Theorem \ref{thm:step_one} and $I_S^*(x)$ in Theorem \ref{thm:N}, we notice that $I_a^*$ in \eqref{eq:I_op}  is a single deductible reinsurance contract with policy limit, but $I_S^*$ in \eqref{eq:opI} consists of $N$ such contracts. We plot the optimal contract $I_S^*$ when $N = 2$ in Figure \ref{fig:s_op}, and it shows that $I_S^*$ is a \emph{multi-layer} reinsurance contract (see Example 4.2 in \citet{jin2024optimal} for a similar result), and $(I_S^*)'$ is either 0 or 1. 
More importantly, $I_a^*(X, S) \le (  a + S)^+ = R$, implying that  $I_a^*$ is a default-free contract. 
However, an endogenous default from the reinsurer is possible if the insurer chooses contract $I_S^*$. To see this, consider the case of $N = 2$ in Figure \ref{fig:s_op} and assume $S = s_1$; the cap on the reinsurer's reserve is $R = (  \pi(I_S^*) + s_1)^+$, but for all $X > l_2 +(  \pi(I_S^*) + s_1)^+$, we easily see from Figure \ref{fig:s_op} that $I_S^*(X) > R$, under which an endogenous default occurs. Also, we obtain  Theorem \ref{thm:step_one} without imposing the IC constraint on indemnities or any assumption on the background risk $S$; in contrast, Theorem \ref{thm:N} requires the IC condition (see $\AcIC$ in \eqref{eq:Acl}) and Assumption \ref{ass2} on $S$. 

\begin{figure}[h]
	\centering
	\vspace{-1ex}
	\begin{tikzpicture}[scale = 0.55]
		\draw[->] (0,0) -- (18,0) node[above, xshift=-5pt] {$x$};
		\draw[->] (0,0) -- (0,10) node[below right] {$I_S^*$};
		
		\foreach \x in {1,5,10,15} 
		\draw (\x,0) node[below] {\ifnum\x=1 $l_1$ \else \ifnum\x=5 \small{$l_1\!+\!(\pi(I_S^*)\!+\!s_1)^+$} \else \ifnum \x=10 \small{$l_2\!+\!(\pi(I_S^*)\!+\!s_1)^+$} \else {\small$l_2\!+\!(\pi(I_S^*)\!+\!s_2)^+$} \fi \fi \fi};
		\foreach \y in {4,9} 
		\draw (0,\y) node[left] {\ifnum\y=4 {\small$(\pi(I_S^*)\!+\!s_1)^+$} \else {\small$(\pi(I_S^*)\!+\!s_2)^+$} \fi};
		
		\draw[domain=0:1, thick] plot (\x,0);
		\draw[domain=1:5, thick] plot (\x,\x-1);
		\draw[domain=5:10, thick] plot (\x, 4);
		\draw[domain=10:15, thick] plot (\x, \x-6);
		\draw[domain=15:17, thick] plot (\x, 9);
		
		\draw[dotted] (5,4) -- (5,0);
		\draw[dotted] (10,4) -- (10,0);
		\draw[dotted] (10,4) -- (0,4);
		\draw[dotted] (15,9) -- (15,0);
		\draw[dotted] (15,9) -- (0,9);
		
		\fill (1, 0) circle (2pt);
		\fill (5, 4) circle (2pt);
		\fill (10, 4) circle (2pt);
		\fill (15, 9) circle (2pt);
		
	\end{tikzpicture}
	\caption{Optimal contract $I^*_S$ in \eqref{eq:opI} when $N = 2$.}
	\label{fig:s_op}
    \vspace{-2ex}
\end{figure}
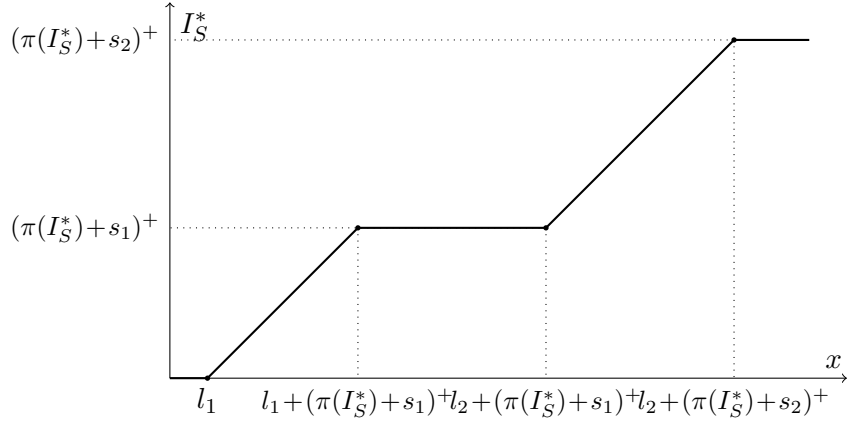

Suppose that $S= r \in \Rb$ (i.e., there is no background risk), and note that this case is covered  under both Theorems \ref{thm:step_one} and \ref{thm:N}. For $N = 1$ in Assumption \ref{ass2}, under which $I_S^*$ becomes a single deductible reinsurance contract with policy limit. In addition, we observe that $l_1$ is the deductible amount of contract $I_S^*$,  and the policy limit is $l_1 +  r + \pi(I_S^*)$. Correspondingly, in Theorem \ref{thm:step_one}, when \( S = r \), \( d(a) \) is the deductible amount, and the policy limit is \( d(a) + r + a \), as shown in \eqref{eq:I_op}. Note that $\pi(I_S^*) = a$ is the premium level.  Thus, for the special case of \( S = r \), the results of Theorem \ref{thm:N} and Theorem \ref{thm:step_one} coincide, as one would expect. However, the two theorems are obtained under \emph{different} admissible sets: Theorem \ref{thm:N} finds the optimal contract over $\AcIC$ defined in \eqref{eq:Acl} (with the IC condition imposed), whereas Theorem \ref{thm:step_one} finds the optimal contract over \( \Ac \) defined in \eqref{eq:Ac}. Because $\AcIC \subset \Ac$ the result of Theorem \ref{thm:step_one} is stronger than those of Theorem \ref{thm:N} when \( N = 1 \). Also, the proof of Theorem \ref{thm:step_one}  is much simpler than that of Theorem \ref{thm:N}.

As implied by the method in Remark \ref{rem:proof}, for every $I \in \AcIC$, we can find an $I_S^*\in \AcIC$ in the form of \eqref{eq:opI} satisfying $\Eb[u(W_S(I_S^*))]\ge\Eb[u(W_S(I))]$; as such, \eqref{eq:opI} provides an \emph{analytical} characterization of the optimal contract. 
In consequence, Theorem \ref{thm:N} reduces Problem \ref{prob:s}, an infinite-dimensional optimization problem over $I \in \AcIC$, into an \( N \)-dimensional optimization problem over $0 \le l_1 \le l_2 \le \cdots \le l_N \le M$. Denoting $\vec{l} = \{l_i\}_{i=1,\cdots,N}$ and $I_S^*$ in \eqref{eq:opI} by $I_S^*(\cdot; \vec{l})$, our next agenda is to solve for the optimal parameters $\vec{l}^*$ defined by 
\begin{align*}
   \vec{l}^* = \argsup_{0\le l_1 \le \cdots \le l_N \le M} 
   \; \Eb \left[ u \left(W_S \big(I_S^*(\cdot; \vec{l})\big) \right) \right].
\end{align*}
However, there is no hope finding $\vec{l}^*$ in a general setup. 
Below, we consider a special case with $N = 2$ and $\tau=1$ and obtain semi-explicit expressions for the \emph{optimal} parameters \( l_1^* \) and \( l_2^* \) under a given premium level. The detailed results are summarized in Proposition \ref{prop:N=2}. 
Applying Theorem \ref{thm:N} to the case of $N=2$, we first obtain the optimal contract in a parametric form by 
\begin{align}
		I_S^*(x; l_1, l_2)&= (x-l_1)^+- \left(x-l_1- \left(   a + s_1 \right )^+ \right)^+ \\
		&\quad + \left(x-l_2- \left(  a + s_1 \right)^+ \right)^+ - \left(x-l_2- \left(  a + s_2 \right)^+ \right)^+, \label{eq:opI_N2}
	\end{align}
	in which $a = \pi(I_S^*) \in [0, \bar{a}_N]$ denotes the premium level, and  $ 0\le l_1 \le l_2 \le M$ are free parameters.

\begin{proposition}\label{prop:N=2}
	Let Assumption \ref{ass2} and Conditions (2) and (3) of Assumption \ref{ass} hold and further assume \(s_N>0\), \( N = 2 \) and \( \tau = 1 \). For a given premium level $a \in [0, \bar{a}_N]$, the optimal reinsurance contract to Problem \ref{prob:s} is given by $I_S^*(\cdot; l_1^*, l_2^*)$ in \eqref{eq:opI_N2}, in which the optimal parameters $l_1^*$ and $l_2^*$ are determined by one of the following cases.	

	\begin{description}
		\item[Case 1:] If $ a+s_1\le 0$, then $l_1^* =l_2^*$ and $l_2^* \in[0,M]$ uniquely solves the following equation:
		$$(1+\eta)\Eb[(X-l_2)^+-(X- l_2 -( a+s_2))^+]-a=0.$$
		
		\item[Case 2:] If $ a+s_1>0$ and \(p_2\le \frac{u'(w - \underline{l_1} - a)}{u'(w - \overline{l_2}- a)}\), in which
			\begin{align*}
			\underline{l_1} &= \inf\{ l_1 \in [0, M] \mid (1 + \eta) \mathbb{E}[(X-l_1)^+-(X-l_1-( a+s_1))^+] - a\le0 \},\\
			\overline{l_2} &=\inf\{l_2 \in[\underline{l_1},M] \mid 
		(1 + \eta) \, \Eb [I_S^*(X; \underline{l_1}, l_2)]
		- a\leq 0 
			\},
		\end{align*}
		then \(l_1^* =\underline{l_1}\) and \(l_2^* =\overline{l_2}\).

		\item[Case 3:] If $ a+s_1>0$ and \( p_2> \frac{u'(w - \underline{l_1} - a)}{u'(w - \overline{l_2}- a)} \), then 
		$(l_1^*, l_2^*)\in[0,M-a-s_1]\times[l_1,M-a-s_1]$ uniquely solves the following equations:
		\begin{align}
			\begin{cases}
			u'(w-l_1-a)=p_2u'(w-l_2-a),\\
			(1 + \eta) \, \Eb[I_S^*(X; l_1, l_2)] = a.
			\end{cases}
		\end{align}
	\end{description}
\end{proposition}
\begin{proof} 
		See Online Companion III.     
\end{proof}

To summarize, Theorem \ref{thm:N} shows that the locally optimal contract to Problem \ref{prob:s} must be in the parametric form of $I_S^*(x; \vec{l})$ in \eqref{eq:opI}, in which $\vec{l} = \{l_i\}_{i=1,\cdots,N}$ is a vector of $N$ free parameters. With additional conditions, Proposition \ref{prop:N=2} finds the optimal parameters $\vec{l}^*$ and thus fully determines the locally optimal contract as $I_S^*(x; \vec{l}^*)$. Recall that a local solution is obtained for a fixed contract premium $a$. For Problem \ref{prob:no}, we go on an extra mile to solve for the optimal premium $a^*$ in Theorem \ref{thm:1_ass}, which leads to the \emph{globally} optimal contracts $I^* = I^*_{a^*}$. Regarding Problem \ref{prob:s}, we can not obtain a similar characterization for $a^*$, at least not without strong assumptions. Nevertheless, assuming a particular distribution for $S$, we can numerically solve for the optimal premium $a^*$, leading to a complete solution to Problem \ref{prob:s}.  We provide one such example below.

\begin{example}\label{exa4}
	Consider an insurer with power utility $u(x) =  x^{1/2}$ and set $\eta=0.1$ (premium loading), 
	$w = 15$ (insurer's initial wealth), and $\tau =1$ (recovery rate). The insurer's loss $X$ follows the same distribution as in Example \ref{exa2}. In addition, the distribution of the reinsurer's background risk (random reserve) $S$ is given by $\Pb(S = s_1 = 2) = 10\%$ and $\Pb(S = s_2 = 8) = 90\%$ (with $N =2$).

For Problem \ref{prob:no}, we use Theorem \ref{thm:1_ass} to compute the optimal premium and deductible, yielding $a^*= \pi(I^*) = 1.00$ and $d(a^*)=4.53$. 

For Problem \ref{prob:s}, we numerically compute the optimal premium and parameters, resulting in $a^* = \pi(I_S^*) = 0.74$, $l_1^*=4.60$, and $l_2^*=6.44$. Recall that the reinsurer's total available reserve is the summation of $S$ and the contract premium $a^*$. For $(x,s)\in[0,10]\times\{2,8\}$, we obtain the optimal contracts $I^*$ and $I^*_S$ by 
\begin{align*}
	I^*(x,s) &= (x-4.53)^+-(x-(5.53+s))^+,  &&(\text{solution to Problem \ref{prob:no}}), \\
	I^*_S(x) &=(x-4.6)^+-(x-7.34)^++(x-9.18)^+ , &&(\text{solution to Problem \ref{prob:s}}).
\end{align*}
It is easy to check that $I^*(X,S)= \Ib(I^*)$; as such, $I^*$ is a default-free contract. However, when $S = s_1 = 2$, $\Ib_{s_1}(I^*_S)=(X-4.6)^+-(X-7.34)^+$, and we have $\Ib_{s_1}(I^*_S) = 2.74 = S + a^* < I_S^*(X)$ for all $X > 9.18$, a scenario corresponding to the reinsurer's endogenous default. The strikingly different behavior between the two optimal contracts $I^*$ and $I^*_S$ highlights the critical impact of feasible contracts on the insurer's reinsurance decision. 

\end{example}

\section{Conclusion} \label{sec:con}
In a one-period economic model,  an insurer with a risk exposure $X$ purchases reinsurance contracts to cover $X$  from a representative reinsurer who applies the expected-value principle to determine premiums.
Motivated by both empirical and theoretical evidence, we model the reinsurer's reserve by a random variable $S$, due to its \emph{background risk}, and take into account the impact of the reinsurer's \emph{endogenous default} on the insurer's reinsurance demand. The insurer seeks an optimal contract that maximizes its expected utility of terminal wealth, subject to the reinsurer's endogenous default and background risk.
First, we consider loss- and background-risk dependent contracts with indemnities in the form of $I(x, s)$ and study the insurer's optimal reinsurance problem in a general setup.
We obtain the (globally) optimal contract in semiclosed form and show that it is a single deductible reinsurance with a policy limit. Next, we consider loss-dependent contracts with indemnities in the form of $I(x)$ and revisit the problem, imposing the IC condition on $I(x)$ and assuming the discrete $S$ is independent of $X$. We obtain an analytical characterization of the (locally) optimal contract, for a fixed premium, and show that it is a multi-layer reinsurance contract.

\section*{Acknowledgment}
\noindent
The first author acknowledges the financial support from the National Natural Science Foundation of China (Grant No.12271290). 

\section*{Competing interests}
\noindent
The authors declare no competing interests. 

\setlength{\bibsep}{0pt plus 0.3ex}
\bibliographystyle{apalike}
\bibliography{references}

@article{arrow1963uncertainty,
	title={Uncertainty and the Welfare Economics of Medical Care},
	author={Arrow, Kenneth J},
	journal={American Economic Review},
	volume={53},
	number={5},
	pages={941--973},
	year={1963},
	publisher={JSTOR}
}

@article{chen2019stochastic,
	title={Stochastic Stackelberg differential reinsurance games under time-inconsistent mean--variance framework},
	author={Chen, Lv and Shen, Yang},
	journal={Insurance: Mathematics and Economics},
	volume={88},
	pages={120--137},
	year={2019},
	publisher={Elsevier}
}

@article{chen2024bowley,
	title={Bowley solution under the reinsurer's default risk},
	author={Chen, Yanhong and Cheung, Ka Chun and Zhang, Yiying},
	journal={Insurance: Mathematics and Economics},
	volume={115},
	pages={36--61},
	year={2024},
	publisher={Elsevier}
}

@article{chi2023optimal,
	title={Optimal risk management with reinsurance and its counterparty risk hedging},
	author={Chi, Yichun and Hu, Tao and Huang, Yuxia},
	journal={Insurance: Mathematics and Economics},
	volume={113},
	pages={274--292},
	year={2023},
	publisher={Elsevier}
}

@article{huberman1983optimal,
	title={Optimal insurance policy indemnity schedules},
	author={Huberman, Gur and Mayers, David and Smith Jr, Clifford W},
	journal={Bell Journal of Economics},
	volume={14},
	number={2},
	pages={415--426},
	year={1983},
	publisher={JSTOR}
}

@article{jin2024optimal,
	title={Optimal moral-hazard-free reinsurance under extended distortion premium principles},
	author={Jin, Zhuo and Xu, Zuo Quan and Zou, Bin},
	journal={SIAM Journal on Control and Optimization},
	volume={62},
	number={3},
	pages={1390--1416},
	year={2024},
	publisher={SIAM}
}

@article{reichel2022optimal,
	title={On the optimal management of counterparty risk in reinsurance contracts},
	author={Reichel, Lukas and Schmeiser, Hato and Schreiber, Florian},
	journal={Journal of Economic Behavior \& Organization},
	volume={201},
	pages={374--394},
	year={2022},
	publisher={Elsevier}
}

@article{mossin1968aspects,
  title={Aspects of rational insurance purchasing},
  author={Mossin, Jan},
  journal={Journal of Political Economy},
  volume={76},
  number={4, Part 1},
  pages={553--568},
  year={1968},
  publisher={The University of Chicago Press}
}

@article{bernard2012impact,
  title={Impact of counterparty risk on the reinsurance market},
  author={Bernard, Carole and Ludkovski, Mike},
  journal={North American Actuarial Journal},
  volume={16},
  number={1},
  pages={87--111},
  year={2012},
  publisher={Taylor \& Francis}
}

@article{biffis2012optimal,
  title={Optimal insurance with counterparty default risk},
  author={Biffis, Enrico and Millossovich, Pietro},
  journal={Available at SSRN 1634883},
  year={2012}
}

@article{cummins2003optimal,
  title={Optimal insurance with divergent beliefs about insurer total default risk},
  author={Cummins, J David and Mahul, Olivier},
  journal={Journal of Risk and Uncertainty},
  volume={27},
  pages={121--138},
  year={2003},
  publisher={Springer}
}

@article{cai2014optimal,
  title={Optimal reinsurance with regulatory initial capital and default risk},
  author={Cai, Jun and Lemieux, Christiane and Liu, Fangda},
  journal={Insurance: Mathematics and Economics},
  volume={57},
  pages={13--24},
  year={2014},
  publisher={Elsevier}
}

@article{levy1994absolute,
  title={Absolute and relative risk aversion: An experimental study},
  author={Levy, Haim},
  journal={Journal of Risk and Uncertainty},
  volume={8},
  pages={289--307},
  year={1994},
  publisher={Springer}
}

@article{schlesinger1981optimal,
  title={The optimal level of deductibility in insurance contracts},
  author={Schlesinger, Harris},
  journal={Journal of Risk and Insurance},
  volume={48},
  number={3},
  pages={465--481},
  year={1981},
  publisher={JSTOR}
}

@incollection{schlesinger2013theory,
	title={The theory of insurance demand},
	author={Schlesinger, Harris},
	booktitle={Handbook of Insurance},
	editor={Dionne, Georges},
	pages={167--184},
	year={2013},
	publisher={Springer}
}

@article{xu2019optimal,
  title={Optimal insurance under rank-dependent utility and incentive compatibility},
  author={Xu, Zuo Quan and Zhou, Xun Yu and Zhuang, Sheng Chao},
  journal={Mathematical Finance},
  volume={29},
  number={2},
  pages={659--692},
  year={2019},
  publisher={Wiley Online Library}
}

@article{asimit2013optimal,
	title={Optimal reinsurance in the presence of counterparty default risk},
	author={Asimit, Alexandru V and Badescu, Alexandru M and Cheung, Ka Chun},
	journal={Insurance: Mathematics and Economics},
	volume={53},
	number={3},
	pages={690--697},
	year={2013},
	publisher={Elsevier}
}

@article{chen2024optimal,
	title={Optimal insurance with counterparty and additive background risk},
	author={Chen, Yanhong},
	journal={ASTIN Bulletin},
	volume={54},
	number={2},
	pages={441--462},
	year={2024},
	publisher={Cambridge University Press}
}

@article{yong2024optimal,
	title={Optimal reinsurance design under distortion risk measures and reinsurer’s default risk with partial recovery},
	author={Yong, Yaodi and Cheung, Ka Chun and Zhang, Yiying},
	journal={ASTIN Bulletin},
	volume={54},
	number={3},
	pages={738--766},
	year={2024},
	publisher={Cambridge University Press}
}

@article{chi2020optimal,
	title={Optimal insurance with background risk: {An} analysis of general dependence structures},
	author={Chi, Yichun and Wei, Wei},
	journal={Finance and Stochastics},
	volume={24},
	number={4},
	pages={903--937},
	year={2020},
	publisher={Springer}
}

@article{doherty1990rational,
  title={Rational insurance purchasing: {Consideration} of contract nonperformance},
  author={Doherty, Neil A and Schlesinger, Harris},
  journal={Quarterly Journal of Economics},
  volume={105},
  number={1},
  pages={243--253},
  year={1990},
  publisher={MIT Press}
}

@article{bernard2015optimal,
  title={Optimal insurance design under rank-dependent expected utility},
  author={Bernard, Carole and He, Xuedong and Yan, Jia-An and Zhou, Xun Yu},
  journal={Mathematical Finance},
  volume={25},
  number={1},
  pages={154--186},
  year={2015},
  publisher={Wiley Online Library}
}

@article{boonen2022mean,
  title={Mean--variance insurance design with counterparty risk and incentive compatibility},
  author={Boonen, Tim J and Jiang, Wenjun},
  journal={ASTIN Bulletin},
  volume={52},
  number={2},
  pages={645--667},
  year={2022},
  publisher={Cambridge University Press}
}

@article{peter2020you,
	title={Do you trust your insurer? {Ambiguity} about contract nonperformance and optimal insurance demand},
	author={Peter, Richard and Ying, Jie},
	journal={Journal of Economic Behavior \& Organization},
	volume={180},
	pages={938--954},
	year={2020},
	publisher={Elsevier}
}

@article{filipovic2015optimal,
  title={Optimal investment and premium policies under risk shifting and solvency regulation},
  author={Filipovi{\'c}, Damir and Kremslehner, Robert and Muermann, Alexander},
  journal={Journal of Risk and Insurance},
  volume={82},
  number={2},
  pages={261--288},
  year={2015},
  publisher={Wiley Online Library}
}

@article{boonen2019equilibrium,
  title={Equilibrium recoveries in insurance markets with limited liability},
  author={Boonen, Tim J},
  journal={Journal of Mathematical Economics},
  volume={85},
  pages={38--45},
  year={2019},
  publisher={Elsevier}
}

@article{doherty1983optimal,
  title={The optimal deductible for an insurance policy when initial wealth is random},
  author={Doherty, Neil A and Schlesinger, Harris},
  journal={Journal of Business},
  volume={56},
  number={4},
  pages={555--565},
  year={1983},
  publisher={JSTOR}
}

@article{mayers1983interdependence,
  title={The interdependence of individual portfolio decisions and the demand for insurance},
  author={Mayers, David and Smith Jr, Clifford W},
  journal={Journal of Political Economy},
  volume={91},
  number={2},
  pages={304--311},
  year={1983},
  publisher={The University of Chicago Press}
}

@article{albrecher2019randomized,
  title={On randomized reinsurance contracts},
  author={Albrecher, Hansj{\"o}rg and Cani, Arian},
  journal={Insurance: Mathematics and Economics},
  volume={84},
  pages={67--78},
  year={2019},
  publisher={Elsevier}
}

@article{asimit2021risk,
  title={Risk sharing with multiple indemnity environments},
  author={Asimit, Alexandru V and Boonen, Tim J and Chi, Yichun and Chong, Wing Fung},
  journal={European Journal of Operational Research},
  volume={295},
  number={2},
  pages={587--603},
  year={2021},
  publisher={Elsevier}
}

@incollection{gollier2013economics,
	title={The economics of optimal insurance design},
	author={Gollier, Christian},
	booktitle={Handbook of Insurance},
	editor={Dionne, Georges},
	pages={107--122},
	year={2013},
	publisher={Springer, New York}
}

@article{cai2020optimal,
	title={Optimal reinsurance designs based on risk measures: {A} review},
	author={Cai, Jun and Chi, Yichun},
	journal={Statistical Theory and Related Fields},
	volume={4},
	number={1},
	pages={1--13},
	year={2020},
	publisher={Taylor \& Francis}
}

@article{borch1962equilibrium,
	title={Equilibrium in a Reinsurance Market},
	author={Borch, Karl},
	journal={Econometrica},
	volume={30},
	number={3},
	pages={424--444},
	year={1962},
	publisher={JSTOR}
}

@article{birghila2023optimal,
  title={Optimal insurance under maxmin expected utility},
  author={Birghila, Corina and Boonen, Tim J and Ghossoub, Mario},
  journal={Finance and Stochastics},
  volume={27},
  number={2},
  pages={467--501},
  year={2023},
  publisher={Springer}
}

@techreport{IAIS2012,
  author       = {{International Association of Insurance Supervisors (IAIS)}},
  title        = {{Reinsurance and Financial Stability}},
  institution  = {{International Association of Insurance Supervisors}},
  type         = {Position paper},
  address      = {c/o Bank for International Settlements, Basel, Switzerland},
  month        = {July},
  day          = {19},
  year         = {2012},
  url          = {http://www.iaisweb.org/}, 
  note         = {© International Association of Insurance Supervisors (IAIS), July 2012. Brief excerpts may be reproduced or translated provided the source is stated.}  
}

@article{li2018optimal,
	title={On the optimal risk sharing in reinsurance with random recovery rate},
	author={Li, Chen and Li, Xiaohu},
	journal={Risks},
	volume={6},
	number={4},
	pages={114},
	year={2018},
	publisher={MDPI}
}

\appendix
\section{Proofs of Section \ref{sec:main}}\label{sub:proof3}
\subsection{Proof of Lemma \ref{lem:a_bar}}\label{prooflm3.1}
\begin{proof} 
	By the definition in \eqref{eq:a_bar}, $g$ is continuous; in addition, $g(0)\ge0$ and $g(\pf) \le 0$.
	Denote $ \mathcal{P} := \left\{ a \in [0, \pf] \mid g(a) \le 0 \right\}\neq\emptyset $. Recall from \eqref{eq:abars} that $ \bar{a} = \inf \mathcal{P} $, and we immediately have $g(a)>0$ for all $a\in[0,\bar{a})$. By the continuity of $g$ and the boundary results, we must have $g(\bar{a}) = 0$.
\end{proof}

\subsection{Proof of Proposition \ref{prop:a_bar}}\label{prop3.1}
\begin{proof} 
	By definition, \( \bar{I}(x,s) = x - (x - (s + \bar{a})^+)^+ \) and $\Ib(\bar{I})=\bar{I}(X,S)$ for all $X \ge 0$ and $S\in\Rb$ (i.e., $\bar{I}$ is a default-free reinsurance contract). For all \( a \in [\bar{a}, \pf] \) and  \( I \in \Ac_a \), we have 
	\begin{align}
		&\quad\Eb[u(W(\bar{I}))]-\Eb[u(W(I))]
		\ge\Eb[u'(w-X+\bar{I}(X,S)-\bar{a})(\bar{I}(X,S)-\Ib(I)+a-\bar{a})]\\
		&=\Eb \left[u'(w-\bar{a})(\bar{I}(X,S)-\Ib(I)+a-\bar{a})\,\mathbf{1}_{\{ X<(S+\bar{a})^+\}} \right]\\
		&\quad +\Eb \left[u'(w-X+(S+\bar{a})^+-\bar{a})((S+\bar{a})^++a-\bar{a}-\Ib(I))\,\mathbf{1}_{\{X\ge (S+\bar{a})^+\}} \right]\\
		&\ge\Eb[u'(w-\bar{a})(\bar{I}(X,S)-\Ib(I)+a-\bar{a})]\ge u'(w-\bar{a})\frac{\eta}{1+\eta}(a-\bar{a}).\label{eq:prop:a_bar}
	\end{align}
	The first inequality is due to $u''<0$, and the next line is obtained by recalling the definition of $\bar{I}$. The second inequality arises from the fact that $u'(w-x+(s+\bar{a})^+-\bar{a})$ is an increasing function of $x$. Additionally, for $a\ge \bar{a}$, $(s+\bar{a})^++a-\bar{a}\ge (s+a)^+$, and $\Ib(I)\le (S+a)^+$ because $(S+a)^+$ is the maximum reserve of the reinsurer for all $I \in \Ac_a$. To get the third inequality, note that $\Eb[\bar{I}(X,S)]= \frac{\bar{a}}{1+\eta}$ and $\Eb[\Ib(I)]\le \frac{a}{1+\eta}$. Therefore, $\Eb[u(W(\bar{I}))]\ge\Eb[u(W(I))]$. 
	
	 If \( S \geq 0 \) and \( \eta = 0 \), then in the second inequality, for all \( a \in [0, \pf] \), \( (S + \bar{a})^+ + a - \bar{a} = (S + a)^+ \). Finally, Proposition \ref{prop:a_bar}  follows. 
\end{proof}

\subsection{Proof of Theorem \ref{thm:step_one}}\label{thm3.1}
\begin{proof} 
	First, we show that for all \(a \in [0, \bar{a}]\), there exists a solution $d = d(a)$ to the equation $g_a(y) = 0$ over $y \in [0, M]$, in which $g_a(y) := (1 + \eta) \Eb[(X - y)^+ - (X - y - (S + a)^+)^+] - a$.
	We deduce $g_a(0) \ge 0$ from Lemma \ref{lem:a_bar} and $g_a(M) = - a \le 0$ as $M = \mathrm{ess \, sup} \, X$. These results, along with the continuity of $g_a$, imply the existence of a solution to $g_a(d) = 0$ over $[0, M]$. 
	
	Motivated by the desirable properties of contract $\bar{I}$, we construct an admissible indemnity  $\tI \in \Ac_a$ in the following form:
	$\tI (x,s) = (x - d(a))^+ - (x - d(a) - (s + a)^+)^+$.
    Because $g_a(d(a)) = 0$, $\pi(\tI) = a$ and $\tI \in \Ac_a$. Using the above definition and \eqref{eq:D}, we easily see that $D(\tI) \equiv 0$, and $\tI$ is a default-free reinsurance contract (i.e., $\Ib(\tI) = \tI(X,S)$). 
	
	For all $a \in [0, \bar{a}]$ and $I \in \Ac_a$,  we obtain 
    \begin{align}
		& \Eb[u(W(\tI))]-\Eb[u(W(I))]
		\ge \Eb \left[u'(w-X+\tI(X,S)-a)(\tI(X,S)-\Ib(I)) \right]\\
		=& \Eb \left[u'(w\!-\!X\!-\!a)(-\Ib(I)) \mathbf{1}_{\{X\le d(a)\}}\right] \\
        & +\Eb \left[u'(w\!-\!d(a)\!-\!a)(\tI(X,S)-\Ib(I)) \mathbf{1}_{\{d(a)<X<d(a)+(S+a)^+\}}\right]\\
		&+ \Eb \left[u'(w-X+(S+a)^+-a)((S+a)^+-\Ib(I)) \mathbf{1}_{\{X\ge d(a)+(S+a)^+\}}\right]\\
		\ge& \, u'(w-d(a)-a) \, \Eb[\tI(X,S)-\Ib(I)]\ge 0,\label{eq:thm:step_one}
	\end{align}
in which all (in)equalities follow from similar arguments as in the proof of Proposition \ref{prop:a_bar}, except the last inequality, which is due to \( \Eb[\tI(X,S)]=\Eb[I(X,S)] \ge \Eb[\Ib(I)] \).  Therefore, we conclude that $\tI = I^*_a$ is the optimal reinsurance contract over the admissible set $\Ac_a$ for all $a \in [0, \bar{a}]$. 
\end{proof}
\subsection{Proof of Corollary \ref{cor:unique_d}}\label{Coroly3.1}
\begin{proof} 
For all \(0\le y_1<y_2\le M\) and $a>0$, we have
\begin{align}
    g_a(y_1)\!-\!g_a(y_2)
    =(1\!+\!\eta)\Eb\big[((X\!-\!y_1)\wedge (S\!+\!a)\wedge(y_2\!-\!y_1)\wedge (S\!+\!a\!+\!y_2\!-\!X))\,\mathbf{1}_{\{y_1<X<y_2+S+a\}}\big] \ge 0.
\end{align}
On $\{y_1 < X < y_2 + S + a\}$, we have $(X-y_1)\wedge (S+a)\wedge(y_2-y_1)\wedge (S+a+y_2-X)>0$. By Condition (3) of Assumption \ref{ass}, we have \( \Pb(y_1 < X < y_2 + S + a)> 0 \). Thus, the above inequality is strict ($g_a$ is strictly increasing for all $a > 0$), and the uniqueness result follows. By Condition (4) of Assumption \ref{ass}, when $a=0$, $d(0)=M$ is the only solution to $g_a(y) = 0$. 
\end{proof}

\subsection{Proof of Theorem \ref{thm:1_ass}}
\label{proof:1_ass}

Based on Theorem \ref{thm:step_one}, the insurer's optimization problem can be transformed into the following one-dimensional optimization problem with respect to the premium \(a\):
\[
\max_{a \in [0, \bar{a}]} \mathbb{E} \left[ u\left( w - X + (X - d(a))^+ - \left(X - (d(a) + S + a)\right)^+ - a \right) \right],
\]
where the deductible amount \(d(a)\) is defined in Corollary \ref{cor:unique_d}. Through continuity and derivative analysis, we obtain a semi-explicit solution for the optimal premium \(a^*\), as shown in Theorem \ref{thm:1_ass}.

To show Theorem \ref{thm:1_ass}, we first present a technical Lemma below regarding the continuity and differentiability of \(d(a)\); its proof is provided in Online Companion I.

\begin{lemma}\label{lem:d_a} 
	Let Assumption \ref{ass} hold. For every $a \in [0, \bar{a}]$, denote $d(a)$ the unique solution to $g_a(y) = 0$ in \eqref{eq:g_a}. Define a set $B_0$ by \(B_0 = \{(a,y)\in[0,\bar{a}]\times[0,M] \mid y \in \mathcal{X}_\Delta \text{ or } y + a \in \mathcal{Z}_\Delta\}\),	and a set $B$ by $B = \{ a \in [0, \bar{a}] \mid (a,d(a))\in B_0\}$, in which sets $\mathcal{X}_\Delta$ and $\mathcal{Z}_\Delta$ includes all the jump points of $X$ and $X-S$ on $[0,M]$. The following two assertions hold:
	\begin{enumerate}
		\item The solution \( d(a) \), as a function of $a$, is continuous on \([0, \bar{a}]\) and continuously differentiable on \((0, \bar{a}] \setminus B\), with the first-order derivative given by
		\begin{align}\label{da}
			d'(a)=\dfrac{(1+\eta) \, \Pb(X>d(a)+S+a)-1}{(1+\eta) \, \Pb(d(a)<X\le d(a)+S+a)}.
		\end{align}
		\item \( B \) is a finite set.
	\end{enumerate}
\end{lemma}

\begin{proof}[Proof of Theorem \ref{thm:1_ass}.]
	From Proposition \ref{prop:a_bar}, we have $I^*(x,s) = \bar{I}(x,s) = x - \left(x - \left(s + \bar{a} \right)\right)^+$ in the case of $\eta = 0$. Also, the same proposition implies that the optimal premium level $a^*$ is achieved on $[0, \bar{a}]$ for all $\eta > 0$. As such, we fix an arbitrary $\eta > 0$ in the rest of the proof. 
	
	From Theorem \ref{thm:step_one}, we know that $I^*_a$ given by \eqref{eq:I_op} is the optimal contract over set $\Ac_a$ and $\pi(I^*_a) = a$, for all $a \in [0, \bar{a}]$. To find the optimal premium level $a^*$, we consider the objective value of contract $I^*_a$, and it equals $\hat{J}(a) :=J(a,d(a))$,  in which
	\begin{align*}
		J(a,y) := \Eb\left[u\left(w - X + (X - y)^+ - (X - (y + S + a))^+ - a\right)\right].
	\end{align*}
	Based on  Lemma \ref{lem:d_a}, the function $\hat{J}$ is continuous on $[0,\bar{a}]$. Using Fubini's theorem, we obtain two equivalent expression of $J$ by 
	\begin{align}\label{eq:J_y}
		J(a,y)=\int_{t\le y}u'(w-t-a)\Pb(X\le t)\dd t-\int_{t>y+a}u'(w-t)\Pb(X-S>t)\dd t+u(w-y-a)\phantom{ee}
	\end{align}
and 
\begin{align}\label{eq:J_a}
	J(a,y)=\Eb[u(w\!-\!X\!-\!a)\mathbf{1}_{\{X\le y\}}]\!-\!\int_{t>y+a}u'(w\!-\!t)\Pb(X\!-\!S>t)\dd t\!+\!u(w\!-\!y\!-\!a)\Pb(X>y).
\end{align}
For all \( (a,y) \in [0, \bar{a}]\times[0,M] \setminus B_0 \), we obtain \( \frac{\partial J}{\partial y} \) from \eqref{eq:J_y} and \( \frac{\partial J}{\partial a} \) from \eqref{eq:J_a} as follows:
\begin{align*}
	\frac{\partial J}{\partial y}(a,y)
	&=-u'(w-y-a)\Pb(y<X\le y+S+a),\\
	\frac{\partial J}{\partial a}(a,y)
	&=-\Eb[u'(w-X-a)\,\mathbf{1}_{\{X\le y\}}]-u'(w-y-a)\Pb(y<X\le y+S+a).
\end{align*}
For all  $a\in(0,\bar{a}] \setminus B$, taking the derivatives of $\hat{J}$ and using \eqref{da}, we obtain
	\begin{align*}
		\hat{J}'(a)
        =&-\Eb \left[u'(w-X-a)\,\mathbf{1}_{\{X\le d(a)\}} \right] +(-d'(a)-1) u'(w-d(a)-a)\Pb(d(a)<X\le d(a)+S+a)\\
        =&-\Eb \left[u'(w-X-a)\,\mathbf{1}_{\{X\le d(a)\}} \right] + u'(w-d(a)-a)\left[\frac{1}{1+\eta}-\Pb (X>d(a)) \right]\\
		=&-\Eb[u'(w-(X\wedge d(a))-a)]+\frac{u'(w-d(a)-a)}{1+\eta}.
	\end{align*}

Define a new function $\bar{J}$ by 
\begin{align}\label{eq:wJ_a}
    \bar{J}(a,y) :=-\Eb[u'(w-(X\wedge y)-a)]+\frac{u'(w-y-a)}{1+\eta},
\end{align}
and let $\widetilde{J}(a) :=\bar{J}(a,d(a))$.
Note that \(\widetilde{J} = \hat{J}'\) on \( (0, \bar{a}] \setminus B \), and \( \widetilde{J} \) is continuous on \([0, \bar{a}]\). Again using Fubini's theorem, we get
\begin{align}\label{eq:wJ_y}
	\bar{J}(a,y)=-\int_{t\le y}u''(w-t-a)\Pb(X\le t)\dd t-\frac{\eta}{1+\eta}u'(w-y-a)
\end{align}
for all \( (a,y) \in [0, \bar{a}]\times([0,M] \setminus \mathcal{X}_{\Delta}) \). Computing \( \frac{\partial J}{\partial y} \) from \eqref{eq:wJ_y} and \( \frac{\partial J}{\partial a} \) from \eqref{eq:wJ_a}, for all  $a\in(0,\bar{a}] \setminus B$, we have
	\begin{align*}
		\widetilde{J}'(a) &=\Eb[u''(w-X-a)\,\mathbf{1}_{\{X\le d(a)\}}]+(d'(a)+1)u''(w-d(a)-a)\left[\Pb(X>d(a))-\frac{1}{1+\eta}\right]\\
		&=\Eb \left[u''(w-X-a)\,\mathbf{1}_{\{X\le d(a)\}}\right]+\frac{u''(w-d(a)-a)\left[1-(1+\eta)\Pb(X>d(a))\right]^2}{(1+\eta)^2 \, \Pb \big( d(a)<X\le d(a)+S+a \big) }.
	\end{align*} 
From \( u'' < 0 \), we deduce \( \widetilde{J}'(a) < 0 \) for \(a\in (0, \bar{a}) \setminus B \). Additionally,  as \(B\) is a finite set, \( \widetilde{J} \) is a strictly decreasing function. On the two boundary points $0$ and $\bar{a}$, we compute 
	\begin{align*}
		\widetilde{J}(0)=-\Eb[u'(w-X)]+\frac{u'(w-M)}{1+\eta}
		\quad \text{and} \quad 
		\widetilde{J}(\bar{a})=\frac{-\eta}{1+\eta}u'(w-\bar{a}) < 0,
	\end{align*}
in which the inequality is due to $\eta > 0$. 
Note that if we take $\eta = 0$, $\widetilde{J}(\bar{a}) = 0$, $J'(a)>0$ for  \(a\in (0, \bar{a}) \setminus B \), and thus $J$ is strictly increasing and reaches its maximum value at $\bar{a}$ (i.e., $a^* = \bar{a}$); in such a case, we easily see that $I^* = \bar{I}$, recovering the result in Proposition \ref{prop:a_bar}. 

With the above results in hand, we discuss two distinctive cases based on the sign of $\widetilde{J}(0)$ and derive the optimal premium $a^*$ in each case accordingly. Thus, the proof is  complete. 
\end{proof}
\subsection {Proof of Proposition \ref{prop:unique}}\label{prop3.2}
\begin{proof} 
First, we show that the inequality in \eqref{eq:ineq_Ibar} is strict. Assume to the contrary that there exists a constant \( a > \bar{a} \) and \( I \in \mathcal{A}_a \)  such that \( \mathbb{E}[u(W(\bar{I}))] =\mathbb{E}[u(W(I))] \). From \eqref{eq:prop:a_bar} in the proof of Proposition \ref{prop:a_bar}, \( \mathbb{E}[u(W(\bar{I}))] =\mathbb{E}[u(W(I))] \) holds if and only if \( \Ib(I) = I(X,S) \) (the last inequality) and \( \bar{I}(X,S) - \bar{a} = \Ib(I) - a \) (the first inequality). From the definition of $\bar{I}$ in \eqref{eq:I_bar}, for \( X < (S + a)^+ \), we have \(\bar{I}(X,S)-\bar{a}+a>X\) and thus \( \mathbb{P}(X \ge (S + a)^+) = 1 \), which contradicts Conditions (1) and (3) of Assumption \ref{ass}. As such, it follows that \( \mathbb{E}[u(W(\bar{I}))] > \mathbb{E}[u(W(I))] \) for all $a > \bar{a}$.

Second, we prove the uniqueness in Theorem \ref{thm:step_one}. That is, for all \( a \le \bar{a} \), and for all \( I \in \mathcal{A}_a \) such that \( \mathbb{E}[u(W(I))] = \mathbb{E}[u(W(I_a^*))] \), it must hold that \( I(X,S) = I_a^*(X,S) \). Recalling the proof of Theorem \ref{thm:step_one}, the inequalities in \eqref{eq:thm:step_one} become equal if and only if \( \Ib(I) = I(X,S) \) (the last inequality) and \( I_a^*(X,S) = \Ib(I) \) (the first inequality). 

Finally, from the proof of Theorem \ref{thm:1_ass}, under Assumption \ref{ass}, for all \( a \le \bar{a} \) and \( a \neq a^* \), we have \( \mathbb{E}[u(W(I^*))]> \mathbb{E}[u(W(I_a^*))] \). 
\end{proof}

\section{Proofs of Section \ref{sec:ext}}
\subsection{Proof of Lemma \ref{lem:a_barN}}\label{prooflm4.2}
\begin{proof} 
Given that $s_N > 0$, it follows that $g_N(0) = (1+\eta)\Eb[X - (X - s_N)^+] > 0$ and $g_N(\pf) = -(1+\eta)\Eb[(X - (s_N + \pf))^+] \le 0$. Let $\bar{a}_N$ be a solution to $g_N(a) = 0$ over $[0,\pf]$. For all $a_1\le a_2$ and $\alpha\in[0,1]$, using the inequality $(a+b)^+\leq a^++b^+$, we have
\begin{align*}
	&g_N(\alpha a_1 + (1 - \alpha)a_2) = (1 + \eta) \mathbb{E} \left[ X - \left( X - (s_N +(\alpha a_1 + (1 - \alpha)a_2)) \right)^+ \right] - (\alpha a_1 + (1 - \alpha)a_2) \\
	&= (1 + \eta) \mathbb{E} \left[ X - \left( \alpha(X - (s_N +a_1)) + (1 - \alpha)(X - (s_N +a_2)) \right)^+ \right] - (\alpha a_1 + (1 - \alpha)a_2) \\
	&\geq (1 + \eta) \mathbb{E} \left[ X - \alpha(X - (s_N +a_1))^+ - (1 - \alpha)(X - (s_N +a_2))^+ \right] - (\alpha a_1 + (1 - \alpha)a_2) \\
	&= \alpha g_N(a_1) + (1 - \alpha) g_N(a_2).
\end{align*}
The concavity of \(g_N\), together with the fact that \(g_N(0) > 0\), implies that \(\bar{a}_N\) is the unique solution to \(g_N(a) = 0\) within \([0, \pi_f]\).
\end{proof}

\subsection{Proof of Theorem \ref{thm:N}}\label{thm4.1}
\begin{proof} 
We first outline the key ideas behind the proof as follows. Denote the insurer's objective function by $\mathcal{J}(\cdot) := \Eb [ u (W_S(\cdot)) ] $, in which $W_S$ is given by \eqref{eq:Ws}. The goal is to show that for all $I \in \AcIC$, there exists an $I_S^*$ in the form of \eqref{eq:opI} such that $\mathcal{J}(I_S^*) \ge \mathcal{J}(I)$. To that end, we fix an arbitrary admissible indemnity \(I \in \AcIC\) and denote its premium \(a := \pi(I) \in [0, \pf]\) in the rest of the proof. 
 
Recall from Assumption \ref{ass2} that the reinsurer's background risk $S$ takes values from $\Sc = \{s_1, \cdots, s_N\}$, and thus the reinsurer's available reserve $R$ takes values from $\{ (a + s_i)^+ \, | \, s_i \in \Sc\}$. These $N$ values of $R$ help partition the loss domain $[0, M]$ into (a maximum of) $N + 1$ sub-intervals $A_i$. We then proceed to complete the proof in two steps. In Step 1, we construct a new indemnity $I_1$ such that, on each sub-interval $A_i$ (i.e., $ X \in A_i$), $I_1(X)$ and $I(X)$ have the same mean, but $I_1$ dominates $I$ in terms of $\mathcal{J}$ (i.e., $\Eb[u(W_S(I_1))\, \mathbf{1}_{\{X \in A_i\}}]\ge \Eb[u(W_S(I))\, \mathbf{1}_{\{X \in A_i\}}]$). This construction, if indeed achieved, immediately shows that $I_1$ is an improvement over $I$ to the insurer. We note that the equal mean constraint helps identify a parameter $l_i$ in $I_1$ for each sub-interval $A_i$, and there are possibly $N + 1$ parameters yet to be determined in the construction of $I_1$. 
In Step 2, we construct $I_2$ based on $I_1$ from Step 1,  which is of the form \eqref{eq:opI}, and show that $I_2$ dominates $I_1$. As such, the optimal contract to Problem \ref{prob:s} must be in the form of $I_2$ in \eqref{eq:opI}. We remark that there are up to $N$ parameters in $ I_2 (=I_S^* )$, but up to $N+1$ parameters in $I_1$.

\medskip
\noindent
\textbf{Step 1.}  As briefly explained above, when $S = s_i$, we have $R = (a + s_i)^+$, and any realization of loss $X$ such that $I(X) > R$ leads to an endogenous default event. This motivates us to define critical points $x_i$ by (we set $\inf \, \emptyset = \infty$ by convention)
	\begin{align}
		x_i := \inf \{x \in [0, M] \, | \, I(x) > ( a + s_i)^+\} \wedge M, \quad i =1, \cdots, N. 
	\end{align}
By \(I(0) = 0\) and the continuity of \(I\), we have \(I(x_i) = ( a + s_i)^+\) for all \(x_i < M\). 
 We denote \(x_0 := 0\) and \(x_{N+1} := M\) and define
$N_0 := \inf\{i \in \{1, 2, \ldots, N+1\} \,|\, x_i = M\}$.
By definition, \(x_{N_0} = M\) holds. Because we do not make any assumptions on the value of each $s_i$, it is possible that $N_0 \le N$, under which $x_{N_0} = x_{N_0 + 1} = \cdots = x_{N + 1} = M$. For all \(i \leq N_0 - 1\), we have \(x_i < M\) and \(I(x_i) = ( a + s_i)^+\) .

Using the points $\{x_0, x_1, \cdots, x_N, x_{N+1}\}$, define $N+1$ sub-intervals $A_i$ by 
\begin{align*}
	A_1 = [x_0, x_1] \quad \text{and} \quad A_i = (x_{i-1}, x_i], \; i = 2, \cdots, N+1.
\end{align*} 
(If $N_0 \le N$, $A_{N_0 + 1}, \cdots, A_{N+1}$ are empty sets.) 
Then, we can partition \([0, M]\) into $\cup_{i=1}^{N_0} A_i$.

	We now construct an alternative admissible indemnity \(I_1 \in \AcIC\) with the same premium as $I$ (i.e., $\pi(I_1) =a$) and show that the insurer prefers $I_1$ to $I$. Denoting \(s_0 := - a\) and \(s_{N+1} := - a + M\) and recalling that $S \in \{s_1, s_2, \cdots, s_N\}$, we define $I_1$ by 
	\begin{align}
		\label{eq:I1_def}
		I_1(x) =\sum_{i=1}^{N_0} \left[ \left(x - l_i^{(1)} - ( a + s_{i-1})^+ \right)^+ - \left(x - l_i^{(1)} - (a + s_i)^+ \right)^+\right],
	\end{align}
	in which $l_i^{(1)}$s are constants yet to be determined. We select constants \(l_i^{(1)} \in [x_{i-1} - ( a + s_{i-1})^+, x_i - ( a + s_i)^+]\) for $i \le N_0 -1$ (or \(l_{N_0}^{(1)} \in [x_{N_0-1} - ( a + s_{N_0-1})^+, x_{N_0} - ( a + s_{N_0-1})^+]\))  so that 
	\begin{align}
		\label{eq:l_eq}
		\mathbb{E} \left[I_1(X) \, \mathbf{1}_{\{X \in A_i\}} \right] = \mathbb{E} \left[I(X) \, \mathbf{1}_{\{X \in A_i\}} \right], \quad i=1, \cdots, N_0,
	\end{align}
    which in turn implies that \(\mathbb{E}[I_1(X)] = \mathbb{E}[I(X)]\) and \(\pi(I_1) = a\).
    
    With the construction of $I_1$ above, Step 1 boils down to the following two tasks:
    \\
    (i) Show that there exist constants $l_i^{(1)}$ such that \eqref{eq:l_eq} holds for all $i =1, \cdots, N_0$. 
    \\
    (ii) Show that $I_1 $ dominates $I$ in terms of $\mathcal{J}(\cdot)$ when $X \in A_i$ for all $i =1, \cdots, N_0$.

In the rest of Step 1, we fix an $i=1, \cdots, N_0$  and focus on the losses that fall in $A_i$ (i.e., $X \in A_i$). 

\medskip 
\noindent
\textbf{Task (i).} 
	Because \(l_{i}^{(1)} + ( a + s_{i})^+ \leq x_i \leq l_{i+1}^{(1)} + ( a + s_i)^+\) for \( 1\le i \leq N_0 - 1\), we obtain, for all $x \in A_i = (x_{i-1}, x_i]$, $2\le i\le N_0$ (or $A_1 = [0,x_1]$), that (see the definition of $I_1$ in \eqref{eq:I1_def})
	\begin{align}\label{eq:I1}
		I_1(x) &= ( a +s_{i-1})^+ + \left(x - l_i^{(1)} - ( a+ s_{i-1})^+ \right)^+ - \left(x - l_i^{(1)} - ( a + s_i)^+ \right)^+ .
	\end{align}
	We plot $I_1(x)$ over $x \in A_i = (x_{i-1}, x_i]\,(2\le i\le N_0-1)$ in Figure \ref{fig:sa_op} to visualize $I_1$. 
	
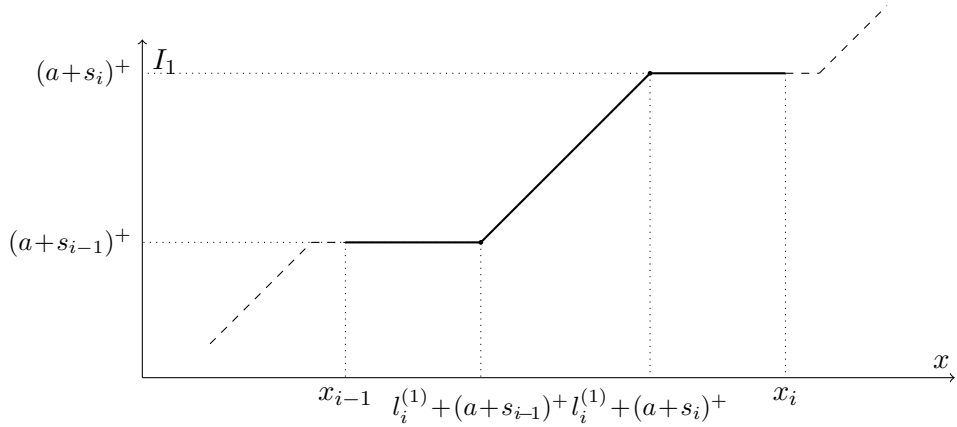
\begin{figure}[h]
	\centering
	\vspace{-1ex}
	\begin{tikzpicture}[scale = 0.45]
		\draw[->] (0,0) -- (24,0) node[above, xshift=-5pt] {$x$};
		\draw[->] (0,0) -- (0,10) node[below right] {$I_1$};
		
		\foreach \x in {6,10,16,19} 
		\draw (\x,0) node[below] {\ifnum\x=6 $x_{i-1}$ \else \ifnum\x=10 {\small$l_i^{(1)}\!+\!(a\!+\!s_{i\!-\!1})^+$} \else \ifnum \x=16 {\small$l_i^{(1)}\!+\!(a\!+\!s_i)^+$} \else  $x_i$ \fi \fi \fi};
		\foreach \y in {4,9} 
		\draw (0,\y) node[left] {\ifnum\y=4 {\small$(a\!+\!s_{i-1})^+$} \else {\small$(a\!+\!s_i)^+$} \fi};
		
		\draw[domain=2:5, dashed] plot (\x,\x-1);
		\draw[domain=5:6, dashed] plot (\x,4);
		\draw[domain=6:10, thick] plot (\x, 4);
		\draw[domain=10:15, thick] plot (\x, \x-6);
		\draw[domain=15:19, thick] plot (\x, 9);
		\draw[domain=19:20, dashed] plot (\x, 9);
		\draw[domain=20:22, dashed] plot (\x, \x-11);
		
		\draw[dotted] (6,4) -- (6,0);
		\draw[dotted] (10,4) -- (10,0);
		\draw[dotted] (10,4) -- (0,4);
		\draw[dotted] (15,9) -- (15,0);
		\draw[dotted] (19,9) -- (19,0);
		\draw[dotted] (15,9) -- (0,9);
		
		\fill (10, 4) circle (2pt);
		\fill (15, 9) circle (2pt);
		
	\end{tikzpicture}
	\caption{Indemnity $I_1$ in \eqref{eq:I1_def} on $x \in (x_{i-1},x_i]$.}
	\label{fig:sa_op}
    \vspace{-2ex}
\end{figure}

Based on the definition of \(x_i\) and the continuity and monotonicity of \(I\), for \(x \in (x_{i-1}, x_i]\) (or \(x \in [0, x_1]\)), we have \(( a + s_{i-1})^+ < I(x) \le ( a + s_i)^+\) (or \(0 \le I(x) \le ( a + s_1)^+\)).
On the one hand, we have
\begin{align*}
	I(x) &\leq \left\{ I(x_{i-1}) + (x - x_{i-1}) \right\} \wedge ( a + s_i)^+\\
	&=( a + s_{i-1})^+ + (x - x_{i-1}) - (x - x_{i-1} + ( a + s_{i-1})^+ - ( a + s_i)^+)^+.
\end{align*}
On the other hand, for all \(i \leq N_0 - 1\),
\begin{align*}
	I(x) &\geq ( a + s_{i-1})^+ \vee (x - x_i + I(x_i)) 
	= ( a + s_{i-1})^+ + (x - x_i + ( a + s_i)^+ - ( a + s_{i-1})^+)^+,
\end{align*}
while for \(i = N_0\), \(I(x) \geq ( a + s_{i-1})^+\). Thus, the existence of such an $l_i^{(1)}$ to \eqref{eq:l_eq} is established for all \(i = 1, \cdots, N_0\). 

\vspace{2ex}
\noindent
\textbf{Task (ii).} For losses $X \in A_i$, the actual indemnity \(\Ib_S(I)\) of contract \(I\) (see its definition in \eqref{eq:Ib_S}) is given by  
	\begin{align*}
		\Ib_{s_j}(I)= \tau( a+s_j)^+, \; j=1, \cdots, i-1,
        \quad \text{and} \quad
		\Ib_{s_j}(I) = I(X), \; j= i, \cdots, N,
	\end{align*}
in which $s_j$ is the realized value of the background risk $S$. 
Similarly, the actual indemnity \(\Ib_S(I_1)\) of contract \(I_1\) in \eqref{eq:I1} is given by 
	\begin{align}
		\Ib_{s_j}(I_1) &= \tau( a+s_j)^+, && j=1, \cdots, i-2,\\
		\Ib_{s_{j}}(I_1) &=( a+s_{i-1})^+ \, \mathbf{1}_{\{X\in(x_{i-1}, \, l_{i}^{(1)}+( a+s_{i-1})^+]\}}&&\\ 
		& \quad +\, \tau( a+s_{i-1})^+ \, \mathbf{1}_{\{X\in(l_{i}^{(1)} +( a+s_{i-1})^+,x_i]\}}, && j = i-1,\\
		\Ib_{s_j}(I_1) &=I_1(X),&& j= i, \cdots, N.
	\end{align}

Comparing \(\Ib_S(I)\) and \(\Ib_S(I_1)\), we easily see that 
\begin{align*}
	\Ib_{s_j}(I) = \Ib_{s_j}(I_1) \text{ for all } j = 1, \cdots, i-2, \quad \text{and} \quad \Ib_{s_{i-1}}(I) \le \Ib_{s_{i-1}}( I_1).
\end{align*}
It then follows that, for all $j=1, \cdots, i-1$, 
\begin{align*}
	\Eb \left[u \left(w-X+ \Ib_{s_j}(I_1) -a \right)\mathbf{1}_{\{X\in A_i\}}\right]\ge \Eb \left[u \left(w-X+\Ib_{s_j}(I)-a \right)\mathbf{1}_{\{X\in A_i\}} \right].
\end{align*}

Next, we consider the cases when $j = i, \cdots, N$. By the above results on $\Ib_{s_j}$ and the construction of $l_i^{(1)}$, we have 
\begin{align*}
	\Eb \left[\Ib_{s_j}(I_1)\mathbf{1}_{\{X\in A_i\}} \right] = \Eb \left[I_1(X)\mathbf{1}_{\{X\in A_i\}}\right] = \Eb \left[I(X)\mathbf{1}_{\{X\in A_i\}}\right] = \Eb \left[\Ib_{s_j}(I)\mathbf{1}_{\{X\in A_i\}}\right].
\end{align*}
In addition, the following (in)equalities hold:
	\begin{align*}
		\Ib_{s_j}(I_1)&=( a+s_{i-1})^+\le \Ib_{s_j}( I),&& \text{if } X \in (x_{i-1},  l_{i}^{(1)} +( a+s_{i-1})^+],\\
		\Ib_{s_j}(I_1)&= X - l_i^{(1)}, && \text{if } X \in (l_{i}^{(1)}\! +\!( a\!+\!s_{i-1})^+, (l_{i}^{(1)} \!+\!( a\!+\!s_{i})^+)\wedge x_i],\\
		\Ib_{s_j}(I_1)&=( a+s_{i})^+\ge \Ib_{s_j}(x; I),&& \text{if } X \in ((l_{i}^{(1)} +( a+s_{i})^+)\wedge x_i, x_{i}].
	\end{align*}
    Note that \(x_{i} < l_{i}^{(1)} +( a+s_{i})^+\) can only possibly hold when \(i = N_0\). Following a similar argument in the proof of Theorem \ref{thm:step_one} and using above results, we obtain, for all $j=i, \cdots, N$, that 
     \begin{align}
		& \Eb \left[u \left(w-X+ \Ib_{s_j}( I_1) -a \right)\mathbf{1}_{\{X\in A_i\}}\right] -  \Eb \left[u \left(w-X+\Ib_{s_j}(I)-a \right)\mathbf{1}_{\{X\in A_i\}} \right]\\
		\ge \, & \Eb \left[u'(w-X+\Ib_{s_j}(I_1)-a) \cdot (\Ib_{s_j}(I_1)-\Ib_{s_j}( I)) \cdot \mathbf{1}_{\{X\in A_i\}}\right]\\
		= \, &\Eb \left[u'(w-X+( a\!+\!s_{i-1})^+-a) \cdot (( a\!+\!s_{i-1})^+-\Ib_{s_j}( I)) \cdot  \mathbf{1}_{\{X\in (x_{i-1}, \, l_{i}^{(1)} +( a+s_{i-1})^+]\}} \right]\\
        &+\Eb \left[u'(w-l_i-a) \cdot  (\Ib_{s_j}(I_1)-\Ib_{s_j}( I)) \cdot  \mathbf{1}_{ \{X \in (l_{i}^{(1)} +( a+s_{i-1})^+, \, (l_{i}^{(1)} +( a+s_{i})^+)\wedge x_i ]\}} \right]\\
		&+\Eb\left[u'(w-X+( a\!+\!s_{i})^+-a) \cdot (( a\!+\!s_{i})^+-\Ib_{s_j}(I)) \cdot \mathbf{1}_{\{X \in ((l_{i}^{(1)}+( a+s_{i})^+)\wedge x_i, \, x_i]\}} \right]\\
		\ge& \, u'(w-l_i-a) \, \Eb \left[(\Ib_{s_j}( I_1)-\Ib_{s_j}( I)) \cdot \mathbf{1}_{\{X\in A_i\}} \right]= 0.\label{eq:thm:N}
	\end{align}

Finally, combining the results for $j \le i-1$ and $j \ge i$, we have 
	\begin{align*}
		\Eb[u(W_S(I_1))]&=\sum_{i=1}^{N_0}\sum_{j=1}^N \Pb(S=s_j) \, \Eb[u(w-X+\Ib_{s_j}( I_1)-a)\,\mathbf{1}_{\{X\in A_i\}}]\\
		&\ge \sum_{i=1}^{N_0}\sum_{j=1}^N \Pb(S=s_j) \, \Eb[u(w-X+\Ib_{s_j}( I)-a)\,\mathbf{1}_{\{X\in A_i\}}]
		=\Eb[u(W_S(I))].
	\end{align*} 
With the completion of both Tasks (i) and (ii), we complete Step 1 of the proof.

\medskip 
\noindent
\textbf{Step 2.} We show that $I_S^*$ in \eqref{eq:opI} outperforms $I_1$ in \eqref{eq:I1_def} in terms of the insurer's objective $\mathcal{J}(\cdot)$. 

Note that there are up to $N$ constants, $l_1, \cdots, l_N$, in $I_S^*$, but up to $N+1$ constants in $I_1$. Indeed, if $N_0 < N + 1$, we can set $l_{N_0+1}^{(1)} = \cdots l_{N + 1}^{(1)} = M$. In this way, we can remove the ``unknown'' $N_0$ in \eqref{eq:I1_def} and write $I_1$ as 
	\begin{align}
		I_1(x) &= 
		\left(x-l_1^{(1)}\right)^+ - \left(x - l_1^{(1)} -(  a +s_1)^+\right)^+ +  \left(x- l_{N+1}^{(1)} - ( a+s_N)^+ \right)^+\\ 
		&\quad +
		\sum_{i=2}^{N}\Big[ \left(x-l_i^{(1)}-(  a +s_{i-1})^+ \right)^+ - \left(x-l_i^{(1)}-(  a +s_i)^+ \right)^+\Big], 
        \label{eq:I1_new}
	\end{align}
in which $a = \pi(I_1)$ denotes the premium of contract $I_1$. Recall that \( \pi(I_1) = \pi(I) = a \), and the parameters \( l_i^{(1)} \) are determined in the first step based on \( I \).

Next, we define a new indemnity function $I_2$ by 
\begin{align}
	I_2(x)
	&= \left(x- l_1^{(2)}\right)^+ - \left(x - l_1^{(2)} -(  a_2 +s_1)^+\right)^+ 
	\\
	&\quad + 
	\sum_{i=2}^{N}
	\left[ \left(x - l_i^{(2)} - (  a_2 + s_{i-1})^+ \right)^+ - \left(x - l_i^{(2)} - (  a_2 + s_{i})^+ \right)^+ \right], \label{eq:I2_def}
\end{align}
in which the constants $l_i^{(2)}$s are defined by 
\begin{align}
	\label{eq:l2}
	l_1^{(2)} = l_1^{(1)} \quad \text{and} \quad l_i^{(2)} =l_i^{(1)}+( a+s_{i-1})^+-(  a_2 +s_{i-1})^+, \; i = 2, \cdots, N,
\end{align}
and $a_2 = \pi(I_2)$ is the premium of contract $I_2$ and takes values in $[0, a]$. Given that \( a_2 \le a \) and \( s_{i-1} \le s_{i} \), we have \( l_i^{(2)} \le l_{i+1}^{(2)} \), implying \( I_2 \in \AcIC \). The definition of $I_2$ in \eqref{eq:I2_def} is \emph{not} complete yet because it is in a parametric form of $a_2$, a \emph{free} parameter in $[0, a]$; this is largely different from the definition of $I_1$ in \eqref{eq:I1_new}, in which $a \in [0, \pf]$ is a \emph{given} constant and equals the premium  $\pi(I)$. 
Denote $I_2(\cdot)$ in \eqref{eq:I2_def} by $I_2(\cdot; a_2)$; the particular $a_2$ we need in \eqref{eq:I2_def} should solve $a_2 = (1 + \eta) \Eb[I_2(X; a_2)]$. It is easy to verify that $(1+\eta) \, \Eb[I_2(X; 0)] \ge 0$ and $(1+\eta) \, \Eb[I_2(X; a)] \le a$, which establishes the existence of a solution $a_2$ to $a_2 = (1 + \eta) \Eb[I_2(X; a_2)]$. By now, $I_2$ in \eqref{eq:I2_def} is well defined.

Note that by setting $l_i = l_i^{(2)}$ and  $s_0 = -r - a_2 \, (= - r - \pi(I_S^*))$,  $I_S^*$ in \eqref{eq:opI} is identical to $I_2$ defined in \eqref{eq:I2_def}. We proceed to show that $I_2$ dominates $I_1$. 
To that end, for every $s_j (=S)$, $j = 1, \cdots, N$, we use \eqref{eq:Ib_S} to derive the actual indemnity of $I_1$ and $I_2$ and, by noting \(( a+s_j)^+-( a_2+s_j)^+\le a-a_2\), obtain 
\begin{align*}
	\Ib_{s_j}(I_1)-\Ib_{s_j}(I_2) &= (I_1(X)-I_2(X))\mathbf{1}_{\{X\le l_{j+1}^{(1)}+( a+s_{j})^+\}}\\
    &\hspace{1em} +(\tau( a+s_j)^+-\tau( a_2+s_j)^+)\mathbf{1}_{\{X>l_{j+1}^{(1)}+( a+s_{j})^+\}}\le a-a_2 .
\end{align*}
Thus, recalling \eqref{eq:Ws}, we have $W_S(I_2) - W_S(I_1) = \Ib_S(I_2)-\Ib_S( I_1)-a_2+a\ge 0$,
implying that $\Eb \left[u(W_S(I_2))\right] \ge \Eb \left[u(W_S(I_1))\right]$ as claimed.  

Finally,  combining the results from Steps 1 and 2, we conclude that the optimal reinsurance contract to Problem \ref{prob:s} is in the form of  $I_S^*$ in \eqref{eq:opI}. Moreover, if $s_N>0$, from $I_S^*(x)\le x\wedge (\pi(I_S^*)+s_N)$, we have
$$\pi(I_S^*)=(1+\eta)\Eb[I_S^*(X)]\le(1+\eta)\Eb[X-(X-(\pi(I_S^*)+s_N))^+]. $$
Recalling Lemma \ref{lem:a_barN}, it follows that $\pi(I_S^*)\in [0,\bar{a}_N]$. 
\end{proof}

\clearpage 
\newpage
\setcounter{page}{1}

\setcounter{section}{0}
\renewcommand\thesection{\Roman{section}}
\renewcommand\thesubsection{\thesection.\arabic{subsection}}

\begin{center}
	\large 
	Online Companion for 
	\\
	``Optimal Reinsurance under Endogenous Default and Background Risk''
	\\
	Zongxia Liang, Zhaojie Ren, and Bin Zou
\end{center}\label{online appendix}

In this Online Companion, we provide technical proofs to Lemma \ref{lem:d_a}, Proposition \ref{prop:a_p}, and Proposition \ref{prop:N=2} in the main paper, along with an extension to the unbounded case where $M=\infty$. Recall that we define $d(a)$ as the solution to \eqref{eq:g_a}, for a given premium $a$, in Theorem \ref{thm:step_one}. In this companion, to avoid potential confusion, we write such a solution by $d (\cdot)$ to emphasize that it is a function defined over $[0, \bar{a}]$, and use $d$ as a generic constant or argument.

\section{Proof of Lemma \ref{lem:d_a}}\label{plem:d_a}

Proof of Item 1. Define  $G:[0,\bar{a}]\times[0,M]\to\Rb$  by
\begin{align}\label{eq:G}
	G(a, y)=(1+\eta)\Eb[(X-y)^+-(X-(y+S+a))^+]-a.
\end{align}
The function \(G\) is continuous on \([0,\bar{a}] \times [0,M]\). Denote the distribution functions of $X$ and $X-S$ by $F_1$ and $F_2$, respectively. Using Fubini's Theorem, we get
\begin{align*}
	\Eb[(X-y)^+]=&\int_{x>y}(x-y)\dd F_1(x)=\int_{x>y}\int_{y<t<x}\dd t\dd F_1(x)
	=\int_{t>y}\Pb(X>t)\dd t.
\end{align*}
Similarly, we have
\begin{align*}
	\Eb[(X-S-y-a)^+]=\int_{t>y+a}\Pb(X-S>t)\dd t.
\end{align*}
Thus, 
\begin{align}\label{gnew}
	G(a, y)=(1+\eta)\left[\int_{t>y}\Pb(X>t)\dd t-\int_{t>y+a}\Pb(X-S>t)\dd t\right]-a.
\end{align}
Recall that \(B_0 = \{(a,y)\in[0,\bar{a}]\times[0,M] \mid y \in \mathcal{X}_\Delta \text{ or } y + a \in \mathcal{Z}_\Delta\}\). From \eqref{gnew}, $G$ has the following partial derivatives on \([0,\bar{a}] \times [0,M] \setminus B_0\):
\begin{align*}
	\frac{\partial G}{\partial a}(a,y) &= (1 + \eta) \mathbb{P}(X > y + S + a) - 1, \\
	\frac{\partial G}{\partial y}(a,y) &= -(1 + \eta) \mathbb{P}(y < X \le y + S + a).
\end{align*}

Let $a_0$ be an arbitrary fixed point in $[0, \bar{a}]$ and recall that $d(a_0)$ is the unique solution to $G(a_0, d(a_0)) = 0$. In particular, \( d(0) = M \), and if \( d(a) = M \), then \( a = 0 \). For $a_0\in(0,\bar{a}]$, by the strict decrease in $y$ and continuity of $G$, for $\epsilon >0$ small enough (say \(\epsilon < \min\{d(a_0), M - d(a_0)\}\)), we have 
	\[
	\lim_{a \to a_0} G(a, d(a_0) - \epsilon) = G(a_0, d(a_0) - \epsilon) > 0,
	\]
	and
	\[
	\lim_{a \to a_0} G(a, d(a_0) + \epsilon) = G(a_0, d(a_0) + \epsilon) < 0.
	\]
If \(d(a_0) = 0\), we only need to consider the second limit with $\epsilon<M$. As such, there exists a positive $\delta$ such that for all \(a\) satisfying \( |a - a_0| < \delta \), we have \( G(a, d(a_0) - \epsilon) > 0 \) and \( G(a, d(a_0) + \epsilon) < 0 \). From the definition of \(d(\cdot)\), it follows that \( d(a) \in (d(a_0) - \epsilon, d(a_0) + \epsilon) \), implying the continuity of \(d(\cdot)\) at \(a_0\). If $a_0=0$, for $0<\epsilon<M$, $G(0,M-\epsilon)>0$. A similar argument shows that \( d \) is continuous at 0. Therefore, \( d (\cdot) \) is continuous on \([0, \bar{a}]\).
	
	Recall that \( B = \{a\in[0,\bar{a}] \mid (a, d(a)) \in B_0\}\). For all $(a,y)\in(0,\bar{a}]\times[0,M)\setminus B_0$, we have $\frac{\partial G(a,y)}{\partial y}<0$  using Condition (3) in Assumption \ref{ass}.
	By the implicit function theorem, $d(\cdot)$ is a continuously differentiable function on $(0,\bar{a}]\setminus B$, and its derivative is given by \eqref{da} as claimed.

	\medskip
	\noindent
	{Proof of Item 2.} For every \( a \in B \), we discuss the following two cases: \( d(a)  \in \mathcal{X}_\Delta\) or \( d(a)  + a  \in \mathcal{Z}_\Delta \).
	
	\textbf{Case 1:}  Let $a$ be such that \( d(a) = y \in \mathcal{X}_\Delta\). Define a function $G_1$ by 
	\[
	G_1(\tilde{a}) = (1 + \eta) \mathbb{E}[(X - y)^+ - (X - (y + S + \tilde{a}))^+] - \tilde{a}.
	\]
If $y=M$, then $a=0$. We only need to consider $y<M$. By the definition of $d (\cdot)$, we have \( G_1(a) = 0 \). Similar to the proof of Lemma \ref{lem:a_barN}, we can show that \( G_1 \) is a continuous and concave function. In addition, \( G_1(0) = (1 + \eta) \mathbb{E}[(X - y)^+ - (X - (y + S))^+] > 0 \) and \( G_1(\pf) \leq 0 \). Hence, $a$ is identified as the unique zero of $G_1$. 
	
	\textbf{Case 2:} Let $a$ be such that \( d(a) + a = y\in\mathcal{Z}_\Delta \).  Define a function $G_2$ by 
	\[
	G_2(\tilde{a}) = (1 + \eta) \mathbb{E}[(X - (y  - \tilde{a}))^+ - (X -( y +S))^+] - \tilde{a}.
	\]
Using the definition of $d (\cdot)$, we identify $a$ as a solution to \( G_2(a) = 0 \), and our remaining task is to show that there are finitely many such solutions. Recall the function $d (\cdot)$ takes values over $[0, M]$, the solutions $a$ are in the interval $[0, y ]$.  For all \( 0 \leq \tilde{a}_1 < \tilde{a}_2 \leq y \) and \( \alpha \in (0, 1) \), we have 
\begin{align*}
G_2(\alpha \cdot \tilde{a}_1 + (1 - \alpha) \cdot \tilde{a}_2) - \alpha \cdot G_2(\tilde{a}_1) - (1-\alpha) \cdot G_2(\tilde{a}_2) =(1+\eta) \Eb \big\{  \mathbf{1}_{\{X \in (y - \tilde{a}_2, y  - \tilde{a}_1)\}} \\
\qquad \cdot \left[\left(\alpha(X - y  + \tilde{a}_1) + (1 - \alpha)(X - y + \tilde{a}_2)\right)^+ - (1 - \alpha)(X - y + \tilde{a}_2)\right] \big\} < 0,
\end{align*}
The last inequality holds because \( y - \tilde{a}_1 > 0 \) and \( y - \tilde{a}_2 < M \). Therefore, \( G_2 \) is a strictly convex function. Consequently, the equation \( G_2(a) = 0 \) on \([0, y]\) has at most two solutions.
	
	Combining the above cases, we conclude that \(B\) is a finite set.   
\section{Proof of Proposition \ref{prop:a_p}}\label{pprop:a_p}
In Proposition \ref{prop:a_p}, we state the comparative statics results of $a^*$, $d^*$, and $U^* = d^* + S + a^*$ in the optimal contract with respect to three model inputs, $S$, $w$, and $\mathbb{A}_u$. Recall that $a^* = \pi(I^*)$ is the contract premium, $d^* = d(a^*)$ is the deductible amount, and $U^*$ is the maximum covered loss; $S$ is the reinsurer's random reserve (background risk), $w$ is the insurer's initial wealth, and $\mathbb{A}_u = - \frac{u''}{u'}$ is the insurer's Arrow-Pratt coefficient of absolute risk aversion. 

By Theorem \ref{thm:1_ass}, when \(\eta \ge \frac{u'(w-M)}{\Eb[u'(w-X)]} - 1\), we have \(a^* = 0\), \(d^* = M\), and \(a^* + d^* = M\). If \( \mathbb{A}_{u}(x) \) is a decreasing function (as in Assumption \ref{asu:u}), taking the first-order derivative shows that \(\frac{u'(w-M)}{\Eb[u'(w-X)]}\) decreases with respect to (w.r.t.) \(w\), the insurer's initial wealth. Given two utility functions \(u_1\) and \(u_2\) such that \(\mathbb{A}_{u_1}(x) \le \mathbb{A}_{u_2}(x)\) for all \(x\), we find that \(\frac{u_1'(x)}{u_2'(x)}\) increases w.r.t. \(x\) and then \(\frac{u_1'(w-M)}{\Eb[u_1'(w-X)]}\le \frac{u_2'(w-M)}{\Eb[u_2'(w-X)]}\). Therefore, all Items in Proposition \ref{prop:a_p} hold when \(\eta \ge \frac{u'(w-M)}{\Eb[u'(w-X)]} - 1\). In the rest of the proof, we only consider the opposite case and make the standing assumption:
\begin{align*}
	\eta < \frac{u'(w-M)}{\Eb[u'(w-X)]} - 1.
\end{align*}
We study each of the three contract specifications $a^*$, $d^*$, and $U^*$, instead of focusing on the three model inputs, one by one. We first study the optimal premium $a^*$ which is the solution to the optimization problem in \eqref{eq:a_op} over all $a \in [0, \pf]$. Then, based on the properties of the optimal premium \( a^* \), we derive the corresponding properties of optimal deductible \( d^*=d(a^*) \) and policy limit \( U^*=a^*+d^*+S\). 

By Theorem \ref{thm:1_ass}, $(a^*, d^*:=d(a^*))$ are determined by jointly solving \eqref{eq:g_a} and \eqref{eq:a_star} over \((a, d) \in [0, \pf] \times [0, M]\). By rewriting \eqref{eq:g_a} and \eqref{eq:a_star},  \((a^*, d^*)\) is the solution to the following system of equations over \((a, d) \in [0, \pf] \times [0, M]\):
\begin{align}
	&\tilde{J}_1(a,d;S):=(1+\eta)\Eb[(X-d)^+-(X-(d+S+a))^+]-a =0,\label{eq:ad_a}\\
	&\tilde{J}_2(a,d;w,u):=-\Eb[u'(w-X-a)\,\mathbf{1}_{\{X\le d\}}]+ u'(w-d-a)\left[\frac{1}{1+\eta}-\Pb(X>d)\right] =0\label{eq:ad_d}.
\end{align}
Let $\tilde{J}_3=\frac{\tilde{J}_2}{u'(w-d-a)}$, or equivalently 
\begin{align}
	\tilde{J}_3(a,d;w,u):=-\frac{\mathbb{E}[u'(w - X \wedge d - a)]}{u'(w - d- a)} + \frac{1}{1 + \eta}.
	\label{eq:J3}
\end{align}
When the relationship between \( \tilde{J}_1 \) and \( S \) is not emphasized, we denote \( \tilde{J}_1(a, d; S) \) simply as \( \tilde{J}_1(a, d) \); similar abbreviations on $\tilde{J}_2$ and $\tilde{J}_3$ will also be used. 

Our first objective is to establish some monotonicity results for $\tilde{J}_i$, $i=1, 2, 3$. To that end, define a constant $\underline{x}_0$ by 
\begin{align*}
	\underline{x}_0 := \inf\left\{x \in [0, M] \,\middle| \,\mathbb{P} ( X > x ) \leq \frac{1}{1 + \eta}\right\}.
\end{align*}
By the definition of $\underline{x}_0$, the right-continuity of \( \Pb (X > \cdot) \), and the strict monotonicity of the distribution function of \( X \) in Assumption \ref{ass}, we can conclude that for all \( x < \underline{x}_0 \), \( \mathbb{P} (X > x) > \frac{1}{1 + \eta} \), for \( x = \underline{x}_0 \), \( \mathbb{P} (X > x) \leq \frac{1}{1 + \eta} \), and for all \( x > \underline{x}_0 \), \( \mathbb{P} (X > x) < \frac{1}{1 + \eta} \). The results are summarized below.

\begin{lemma}\label{lem:mono} 
	Let $\tilde{J}_1$, $\tilde{J}_2$, and $\tilde{J}_3$ be defined by \eqref{eq:ad_a}, \eqref{eq:ad_d}, and \eqref{eq:J3}, respectively. Then
	\begin{enumerate}
		\itemsep0em 
		\item For every fixed \( d \in [\underline{x}_0, M] \), \( \tilde{J}_1 (a, d) \) strictly decreases w.r.t. \( a \).
		\item  For every fixed \(a\), \( \tilde{J}_1 (a, d) \) decreases w.r.t. \( d \), and there exists a unique \( d(a) \) such that \( \tilde{J}_1(a, d(a)) = 0 \).
		\item For every fixed \((a, d)\), \( \tilde{J}_1 (a, d; S) \) increases w.r.t. \( S \).
		\item For every fixed \(a\), \( \tilde{J}_2(a, d)\) strictly increases w.r.t. $d$ on $[\underline{x}_0, M]$.
		\item Under Assumption \ref{asu:u}, for every fixed \((a, d)\), \( \tilde{J}_3 (a, d; w)\) decreases w.r.t. $w$.
		\item Under Assumption \ref{asu:u}, for every fixed \(d\), \( \tilde{J}_3(a, d) \) increases w.r.t. $a$.
	\end{enumerate}
\end{lemma}
\begin{proof} 
Differentiating \( \tilde{J}_1 \) with respect to \( a \), as shown in the proof of Lemma \ref{lem:d_a}, proves Item 1. Item 2 follows from Corollary \ref{cor:unique_d}. Item 3 is clearly true. Differentiating \( \tilde{J}_2 \) with respect to \( d \), as in the proof of Theorem \ref{thm:1_ass}, verifies Item 4. For Item 5, differentiating \( \tilde{J}_3 \) w.r.t. \( w \) gives
\[
\frac{\partial \tilde{J}_3}{\partial w} = \frac{-\mathbb{E}[u''(w - X\wedge d- a) ]}{u'(w - d - a)} + \frac{\mathbb{E}[u'(w - X\wedge d - a) ] u''(w - d - a)}{(u'(w - d - a))^2}.
\]
As \( -\frac{u''}{u'} \) is a decreasing function by Assumption \ref{asu:u}, we have 
\[
-\frac{u''(w - d - a)}{u'(w - d - a)} \geq -\frac{u''(w - X\wedge d - a)}{u'(w - X\wedge d - a)},
\]
which implies \( \frac{\partial \tilde{J}_3}{\partial w} \leq 0 \). Using a similar argument, Item 6 follows. 
\end{proof}
\begin{proof}[ Proof of Proposition \ref{prop:a_p}]
	From \eqref{eq:ad_d}, we know that \(\frac{1}{1+\eta}-\Pb (X>d^*) \ge 0\), in which $d^*=d(a^*)$ is the deductible of the optimal contract $I^*$, and this inequality implies that $d^* \ge \underline{x}_0$. Recalling \eqref{eq:ad_a} and Item 1 of Lemma \ref{lem:mono}, there exists a unique solution to $\tilde{J}_1(a, \underline{x}_0)=0$, which we denote by $\hat{a}$. By Item 1 and Item 2 of Lemma \ref{lem:mono}, it follows that \(a^*\le \hat{a}\). The discussion so far allows us to restrict the feasible domain of \((a, d)\) from \([0, \pf] \times [0, M]\) to \([0,\hat{a}] \times [\underline{x}_0, M]\) in the rest of the proof when we study the properties of the optimal contract.

	\emph{Property 1: \(a^*\) increases w.r.t. \(S\).} Consider two arbitrary background risks, \( S_1 \) and \( S_2 \), that satisfy Assumption \ref{ass} and \( S_1 \le S_2 \). Let \( \hat{a}_i := \hat{a}(S_i) \) denote the unique solution to \(\tilde{J}_1(a, \underline{x}_0;S_i)=0\) when $S = S_i$ for $i = 1, 2$.  It follows from Items 1 and 3 of Lemma \ref{lem:mono} that \(\hat{a}_2\ge\hat{a}_1\). Recalling \eqref{eq:ad_a}, for each \( a \in [0, \hat{a}_i] \), there exists a unique \( d_i(a)\in[\underline{x}_0,M] \) such that 
	\(
	\tilde{J}_1(a,d_i(a);S_i)=0.
	\)
	By Item 2 and Item 3 of Lemma \ref{lem:mono}, it follows that \( d_1(a) \leq d_2(a) \), \(a\in[0,\hat{a}_1]\). Recalling \eqref{eq:ad_d}, \( a=a_i^* \in [0, \hat{a}_i] \) is the unique solution to $\tilde{J}_2(a,d_i(a))=0$. By Item 4 of Lemma \ref{lem:mono}, we get \( \tilde{J}_2(a_1^*,d_2(a_1^*)) \geq \tilde{J}_2(a_1^*,d_1(a_1^*)) = 0 \). From the proof of Theorem \ref{thm:1_ass}, we know that $\tilde{J}_2(a,d_2(a))$ strictly decreases w.r.t. \( a \) on \([0,\hat{a}_2]\). As such, we conclude that \( a_2^* \geq a_1^* \) and Property 1 holds.

	\emph{Property 2: \(a^*\) decreases w.r.t. \(w\).}
		Let two arbitrary initial wealth levels, $w_1 < w_2$, be given, and we denote the corresponding optimal premiums by \( a_i^* \), \(i=1, 2\). Note that \( a_i^* \in [0, \hat{a}] \) is the unique solution to $\tilde{J}_2(a, d(a);w_i)=0$ when $w = w_i$. Using Item 5 of Lemma \ref{lem:mono}, we obtain \( \tilde{J}_3(a_1^*,d(a_1^*);w_2) \leq \tilde{J}_3(a_1^*,d(a_1^*);w_1) = 0 \), implying $ \tilde{J}_2(a_1^*,d(a_1^*);w_2)\le 0$. Because \(\tilde{J}_2(a,d(a);w_2)\) strictly decreases w.r.t. \( a \), we conclude that \( a_2^* \leq a_1^* \), which proves Property 2.

		\emph{Property 3: \(a^*\) increases w.r.t. $\mathbb{A}_u$.} 
			Let two arbitrary utility functions \(u_1\) and \(u_2\) be given such that \(\mathbb{A}_{u_1}(x) \le \mathbb{A}_{u_2}(x)\) for all \(x\), and we denote the corresponding optimal premiums by \( a_i^* \), \( i = 1, 2 \). Note that \( a_i^* \in [0, \hat{a}] \) is the unique solution to $\tilde{J}_2(a,d(a);u_i)=0$ when the utility function is $u = u_i$. Because \(\frac{u_1'(x)}{u_2'(x)}\) increases w.r.t. \(x\), we have
		\begin{align*}
			\frac{u_1'(w-d(a_1^*)-a_1^*)}{u_2'(w-d(a_1^*)-a_1^*)}\le \frac{u_1'(w-X\wedge d(a_1^*)-a_1^*)}{u_2'(w-X\wedge d(a_1^*)-a_1^*)},
		\end{align*}
		which implies $\tilde{J}_2(a_1^*,d(a_1^*);u_2)\ge 0$.
		Because \(\tilde{J}_2(a,d(a);u_2) \) strictly decreases w.r.t. \( a \),  it follows that \( a_2^* \ge a_1^* \) and Property 3 holds.

		\emph{Property 4: $d^*$ decreases w.r.t. $S$.} 
		Recalling \eqref{eq:ad_d} and using Item 4 of Lemma \ref{lem:mono},  \( d^*\in[\underline{x}_0,M] \) is the unique solution to $\tilde{J}_2(a^*, d)=0$. By Items 4 and 6 of Lemma \ref{lem:mono},  the impact of $S$ on $d^*$ is exactly the opposite of that on $a^*$, which together with Property 1, confirms Property 4.

		\emph{Properties 5 and 6: $d^*$ increases w.r.t. $w$ and decreases w.r.t. $\mathbb{A}_u$.} 
		The argument above, together with Items 1 and 2 of Lemma \ref{lem:mono}, prove these two properties.

Finally, we focus on the policy limit (maximum covered loss) $U^*$ of the optimal contract. Because  $U^* = a^* + d^* + S$, we study the properties of $a^* + d^*$. We perform a change of variable from $(a, d)$ to $(a, v) := (a, a + d)$ and denote $v^* := a^* + d^*$. Recalling Lemma \ref{lem:d_a}, for all $a\in(0,\hat{a}]\setminus B$, $$d'(a)+1=\frac{(1+\eta) \Pb(X>d(a))-1}{(1+\eta)\Pb(d(a)<X\le d(a)+S+a)}\le 0.$$ Therefore, $v^*\in[\hat{a}+ \underline{x}_0,M]$. By rewriting \eqref{eq:ad_a} and \eqref{eq:ad_d}, $(a^*,v^*)$ is the solution to the following system of equations over $(a,v)\in[0,\hat{a}]\times[\hat{a}+ \underline{x}_0,M]$:
\begin{align*}
	\tilde{J}_4(a,v):=&(1+\eta)\Eb[(X-(v-a))^+-(X-(v+S))^+]-a=0,\\
	\tilde{J}_5(a,v):=&-\Eb[u'(w - X\wedge (v-a) - a)] + \frac{u'(w - v)}{1 + \eta} = 0.
\end{align*}

\emph{Property 7: \(a^* + d^*\) increases w.r.t. $S$.}
For all $(a,v)\in\{(a,v)\in[0,\hat{a}]\times[\hat{a}+d_0,M]\mid v-a\notin\mathcal{X}_\Delta\}$, differentiating the function \(\tilde{J}_5 \) gives
\begin{align*}
	&\frac{\partial\tilde{J}_5}{\partial a}(a,v)=\Eb[u''(w-X-a)\,\mathbf{1}_{\{X\le v-a\}}]\le 0,\\
	&\frac{\partial\tilde{J}_5}{\partial v}(a,v)=u''(w-v)\left[\Pb(X>v-a)-\frac{1}{1+\eta}\right]\ge 0.
\end{align*}
Given the optimal premium $a^*$, $v^*$ is the solution to $\tilde{J}_5(a^*,v)=0$. Therefore, $v^*$ reacts to the change of $S$ in the same direction as $a^*$, which, combining with Property 1, proves this property.

\emph{Properties 8 and 9: \(a^* + d^*\) increases w.r.t. \(w\) and decreases w.r.t. \(\mathbb{A}_u\).} 
We directly compute $\partial\tilde{J}_4 / \partial a$ and $\partial\tilde{J}_4 / \partial v$. The argument above then proves these two properties. 

Proof of Proposition \ref{prop:a_p} is then complete by Properties 1-9.  
\end{proof}

\section{Proof of Proposition \ref{prop:N=2}}\label{pprop:N=2}
\begin{proof} 
Let $\bar{a}_2$ be defined as $\bar{a}_N$ in Lemma \ref{lem:a_barN} when $N = 2$.
	By Theorem \ref{thm:N}, for all \(a\in[0,\bar{a}_2]\),
	the optimal reinsurance contract must be in the form of 
	\[I_S^*(x;l_1,l_2)=(x-l_1)^+-(x-l_1-( a+s_1)^+)^++(x-l_2-( a+s_1)^+)^+-(x-l_2-( a+s_2))^+,\]
	in which $l_1$ and $l_2$ satisfy $0 \le l_1 \le l_2 \le M$ and $(1+\eta)\Eb[I_1(X;l_1,l_2)]-a=0$. 
	Define $\phi$ by 
	\begin{align}
		\phi(l_1, l_2) = (1+\eta)\Eb[I_S^*(X;l_1,l_2)]-a, \quad  (l_1, l_2) \in [0, M] \times [l_1, M].
		 \label{eq:phi}
	\end{align}
	For all $a\in[0,\bar{a}_2]$, the optimization problem that yields the optimal parameters $(l_1^*, l_2^*)$ is
	\begin{align}\label{eq:op_p1}
		(l_1^*,l_2^*)=\argsup_{(l_1,l_2)\in{\mathcal{L}}} \, \Eb[u(W_S(I_S^*(X;l_1,l_2)))],
	\end{align}
	in which $\mathcal{L} :=\{(l_1,l_2)\in[0, M] \times [l_1, M]\mid\phi(l_1,l_2)=0\}\neq \emptyset.$ In other words, $I_S^*(x; l_1^*, l_2^*)$ is the (locally) optimal contract to Problem \ref{prob:s} for a given premium $a = \pi(I_S^*) \in [0, \bar{a}_2]$.
	
	We prove Proposition \ref{prop:N=2} under three exclusive and exhaustive conditions.

	\textbf{Condition 1:} $ a+s_1\le 0$. Under this condition, $I_S^*$ is independent of $l_1$, and $\mathcal{L}=[0,\tilde{l}_2]\times\{\tilde{l}_2\}$, in which $\tilde{l}_2\in[0,M]$ is the unique solution to $(1 + \eta) \mathbb{E}[ (X - l_2)^+- (X - l_2 - (  a + s_2))^+] - a=0$. Naturally, $l_2^*=\tilde{l}_2$ and $l_1^*$ can take any value in $[0, l_2^*]$. Without loss of generality, we set $l_1^* = l_2^*$. This proves Case 1 in the proposition.

	\textbf{Condition 2:} $ a+s_1>0$ and $(1+\eta)\mathbb{E}[(X-(M-( a+s_1))^+)^+]-a\ge 0$. Under this condition, for all $(l_1,l_2)\in[0,(M-( a+s_1))^+)\times[l_1,M]$, we have
	\begin{align*}
		\phi(l_1,l_2)\ge \phi(l_1,M)>(1 + \eta) \mathbb{E}[(X - (M-( a+s_1))^+)^+]\ge 0.
	\end{align*}
	Thus, $\mathcal{L} \subset[(M-( a+s_1))^+,M]\times[l_1,M]$. For all $(l_1,l_2)\in[(M-( a+s_1))^+,M]\times[l_1,M]$, $I_S^*$ is independent of $l_2$, and so  $\mathcal{L}=\{\tilde{l}_1\}\times[\tilde{l}_1,M]$, in which $l_1=\tilde{l}_1\in[(M-( a+s_1))^+,M]$ is the unique solution to $(1+\eta)\Eb[(X-l_1)^+]-a=0$. Therefore, $l_1^*=\tilde{l}_1$ and $l_2^* \in [\tilde{l}_1,M]$. Without loss of generality, we set $l_2^*=\tilde{l}_1$. Recalling the definition of $\underline{l_1}$ and $\overline{l_2}$, we have $\underline{l_1}=\tilde{l}_1$ and $\overline{l_2}=\tilde{l}_1$. This proves Case 2 in the proposition when $(1+\eta)\mathbb{E}[(X-(M-( a+s_1))^+)^+]-a\ge 0$ holds.

	\textbf{Condition 3:} $0< a+s_1<M$, $(1+\eta)\mathbb{E}[(X-(M-( a+s_1)))^+]-a<0$.
	By the definition of $\phi$ in \eqref{eq:phi}, we have $\phi(l_1,M)\le\phi(l_1,l_2)\le \phi(l_1,l_1)$, in which
	\begin{align*}
		\phi(l_1,l_1) &=(1 + \eta) \mathbb{E}[ (X - l_1)^+- (X - l_1 - (  a + s_2))^+] - a,\\
		\phi(l_1,M)  &=(1 + \eta) \mathbb{E}[ (X - l_1)^+- (X - l_1 - (  a + s_1))^+] - a.
	\end{align*}
	For all $a\in[0,\bar{a}_2]$, $\phi(l_1,l_1)=0$ has a unique solution $l_1=\overline{l_1}\in [0,M]$; recall  $\underline{l_1}=\inf\{l_1\in[0,M] \, | \, \phi(l_1,M)\le 0\}\le\overline{l_1} $. If $(1 + \eta) \mathbb{E}[ X- (X - (  a + s_1))^+] - a\le 0$, then $\underline{l_1}=0$. If $(1 + \eta) \mathbb{E}[ X- (X - (  a + s_1))^+] - a>0$, then $\phi(\underline{l_1},M)=0$. Thus,
	$$\mathcal{L} \subset\{(l_1,l_2)\in[0,M]\times[l_1,M]\,|\,\phi(l_1,M)\le 0\le \phi(l_1,l_1)\}=[\underline{l_1},\overline{l_1}]\times[l_1,M].$$
	From $(1+\eta)\mathbb{E}[(X-(M-( a+s_1)))^+]-a<0$, we have $\overline{l_1}<M-( a+s_1)$, and then $\mathcal{L}$ can be restricted to $[\underline{l_1},\overline{l_1}]\times[l_1,M-( a+s_1)]$. Actually, for all $(l_1,l_2)\in[\underline{l_1},\overline{l_1}]\times[M-( a+s_1),M]$, both $I_S^*$ and $\phi$ are independent of $l_2$. Furthermore, for all $l_1\in[\underline{l_1},\overline{l_1}]$, there exists a unique $l_2(l_1)\in[l_1,M-( a+s_1)]$ such that $\phi(l_1,l_2(l_1))=0$. The optimization problem in \eqref{eq:op_p1} reduces to 
	\begin{align}\label{eq:op2}
		l_1^*=\argsup_{l_1\in[\underline{l_1},\overline{l_1}]}\Eb[u(W_S(I_S^*(\cdot;l_1,l_2(l_1))))].
	\end{align}
	
	Denote \(C_0=\{(l_1,l_2)\in[\underline{l_1},\overline{l_1}]\times[l_1,M-( a+s_1)]\mid l_1\in \mathcal{X}_\Delta \text{ or } l_1+ a+s_1\in \mathcal{X}_\Delta \text { or } l_2+ a+s_1\in \mathcal{X}_\Delta \text{ or } l_2+ a+s_2\in \mathcal{X}_\Delta \}\). The function \( \phi \) is continuously differentiable on \([\underline{l_1},\overline{l_1}]\times[l_1,M-( a+s_1)] \setminus  C_0\), with the partial derivatives given by
	\begin{align*}
		&\frac{\partial \phi(l_1, l_2)}{\partial l_1} = -(1 + \eta) \mathbb{P}(l_1 < X \leq l_1 +   a + s_1),\\
		&\frac{\partial \phi(l_1, l_2)}{\partial l_2} = -(1 + \eta) \mathbb{P}(l_2 +   a + s_1 < X \leq l_2 +   a + s_2).
	\end{align*}

	Denote \( C_1 = \{ l_1 \in [\underline{l_1},\overline{l_1}] \mid (l_1, l_2(l_1)) \in C_0 \text{ or } l_2(l_1)=M-( a+s_1)\} \). Because \(\phi(l_1, l_2)\) strictly decreases w.r.t. \(l_1\) on \([0, M]\) and w.r.t. \(l_2\) on \([0, M - (  a + s_1)]\), and \(\phi\) is continuous, it follows that \(C_1\) is a finite set, and \(l_2(l_1)\) is continuous on \([\underline{l_1}, \overline{l_1}]\). For \((l_1, l_2) \in [\underline{l_1},\overline{l_1}]\times[l_1,M-( a+s_1)) \setminus C_0 \), we have \( \frac{\partial \phi(l_1, l_2)}{\partial l_2} < 0 \).  By the implicit function theorem, \( l_2(l_1) \) is a continuously differentiable function on \( [\underline{l_1},\overline{l_1}]\setminus C_1\), with the first-order derivative given by 
	\[
	l_2'(l_1) = -\frac{\mathbb{P}(l_1 < X \leq l_1 +   a + s_1)}{\mathbb{P}(l_2(l_1) +   a + s_1 < X \leq l_2(l_1) +   a + s_2)}. 
	\]
	
	Recalling \eqref{eq:op2}, the objective function of the insurer is
	\begin{align*}
		\psi(l_1)=\Eb[u(W_S(I_S^*))]
		&=p_1\Eb[u(w-X+(X-l_1)^+-\big(X-l_1-( a+s_1)\big)^+-a)]\\
		& \quad +p_2\Eb[u(w-X+(X-l_1)^+-\big(X-l_1-( a+s_1)\big)^+\\
		& \quad +\big(X-l_2(l_1)-( a+s_1)\big)^+
		+\big(X-l_2(l_1)-( a+s_2)\big)^+-a)],
	\end{align*}
	which is continuous on $[\underline{l_1},\overline{l_1}]$. For $l_1\in[\underline{l_1},\overline{l_1}]\setminus C_1$, taking the derivative, we obtain
	\begin{align}
		\psi'(l_1)=[-u'(w-l_1-a)+p_2u'(w-l_2(l_1)-a)]\Pb(l_1<X\le l_1+ a+s_1).
	\end{align}
	Let \(\widetilde{\psi}:[\underline{l_1},\overline{l_1}]\to \Rb\) be defined by
	\begin{align*}
		\widetilde{\psi}(l_1)=-u'(w-l_1-a)+p_2u'(w-l_2(l_1)-a),
	\end{align*}
    and note that 
	\(\widetilde{\psi}\) is continuous on \([\underline{l_1},\overline{l_1}]\). For $l_1\in[\underline{l_1},\overline{l_1}]\setminus C_1$, 
	\begin{align*}
		\widetilde{\psi}'(l_1)=[u''(w-l_1-a)-p_2u''(w-l_2(l_1)-a)l_2'(l_1)]<0.
	\end{align*}
	Furthermore, the boundary conditions are
	\begin{align*}
		&\widetilde{\psi}(\underline{l_1})=-u'(w-\underline{l_1}-a)+p_2u'(w-l_2(\underline{l_1})-a)~~\text{and}~~\widetilde{\psi}(\overline{l_1})=-p_1u'(w-\overline{l_1}-a)<0.
	\end{align*}
	Recalling the definition of $\overline{l_2}$, we have $\overline{l_2}=l_2(\underline{l_1})$. Therefore, this proves Case 2 when $(1+\eta)\mathbb{E}[(X-(M-( a+s_1)))^+]-a<0$, and Case 3 in the proposition. 
	
The proof is now complete, and the results in Cases 1 to 3 in Proposition \ref{prop:N=2} hold. 
\end{proof}

\section{Extensions}
\subsection{Extension to the Unbounded Loss}\label{app:extension}

Recall that we assume the loss $X$ is bounded above by $M \in (0, \infty)$ in the main paper. In this section, we relax this assumption and discuss the extension to the unbounded case. Because we consider the expected-value premium principle (see \eqref{eq:pi}), we still require $\mathbb{E}[X] < \infty$. 

By extending the analysis from a bounded value $M$ to the limit as $M$ tends to infinity, we obtain results analogous to Theorem \ref{thm:step_one}, Theorem \ref{thm:N}, and Proposition \ref{prop:N=2}, but the deductible amount $d(a)$ now takes values over $[0, \infty]$. Before we present the main result, we first modify Assumption \ref{ass} to suit the unbounded loss and state the new assumption below. 
\begin{assumption}\label{ass_unbounded}
	The insurer's loss $X$ and the reinsurer's background risk $S$ satisfy the following conditions:  
	(1) $S \ge 0$ almost surely (i.e., $\mathbb{P}(S \ge 0) = 1$);  
	(2) both $X$ and $X - S$ have finitely many jump points on $[0, \infty)$;  
	(3) $\mathbb{P}(X \le x)$ is strictly increasing in $x \in [0, \infty)$;  
	(4) $\mathbb{E}[(X - y)^+ - (X - y - S)^+] > 0$ for all $y \in [0, \infty)$;
	(5) the derivative of the utility function $u'(x)$ exists as $x \to -\infty$, denoted by $u'(-\infty) \in [0, \infty]$, and it holds that $\mathbb{E}[u'(w - X - \pi_f)] < \infty$.
\end{assumption}

We now state the main result when the loss is unbounded. The proof is similar to that of Theorem \ref{thm:1_ass}, and we omit it here.

\begin{theorem}
	Let Assumption \ref{ass_unbounded} hold. The globally optimal reinsurance contract $I^*$ for Problem \ref{prob:no} is given by 
	\begin{align*}
		I^*(x,s) = 
		\begin{cases}
			x - (x - (s + \bar{a}))^+, & \text{if } \eta = 0, \\
			(x - d(a^*))^+ - (x - (d(a^*) + s + a^*))^+, & \text{if } 0 < \eta < \dfrac{u'(-\infty)}{\mathbb{E}[u'(w - X)]} - 1, \\
			0, & \text{if } \eta \ge \dfrac{u'(-\infty)}{\mathbb{E}[u'(w - X)]} - 1,
		\end{cases}
	\end{align*}
	in which $d(a)$ is defined in Corollary \ref{cor:unique_d} for all $a \in [0, \bar{a}]$, and $a^* \in (0, \bar{a})$ is the unique solution to 
	\begin{align*}
		\mathbb{E}\left[u'(w - (X \wedge d(a)) - a)\right] - \frac{u'(w - d(a) - a)}{1 + \eta} = 0.
	\end{align*}
\end{theorem}

If $\frac{u'(-\infty)}{\mathbb{E}[u'(w - X)]} = \infty$, the third case (where the premium is zero) disappears, and the comparative statics in Proposition \ref{prop:a_p} continue to hold. If $\frac{u'(-\infty)}{\mathbb{E}[u'(w - X)]} < \infty$, because $\frac{u'(-\infty)}{\mathbb{E}[u'(w - X)]}$ increases in $w$, the comparative statics with respect to $w$ in Proposition \ref{prop:a_p} no longer hold, while the other two results remain valid.

\subsection{Extension to $S=\mathrm{VaR}(I(X))$}

In Section \ref{sec:ext}, we solve for the optimal loss-dependent indemnity $I^*(X)$ when the background risk $S$ is independent of the loss $X$ (see Assumption \ref{ass2}). In this section, we consider a case in which $S$ \emph{depends} on $X$. More specifically, we follow \cite{asimit2013optimal} and \cite{cai2014optimal} and set  $S=\mathrm{VaR}_\alpha(I(X))$ for some threshold level $\alpha\in (0,1)$. We state the main result for this extension below. 

\begin{theorem}
    Assume that $S=\mathrm{VaR}_{\alpha}(I(X))$ for a given threshold level $\alpha\in (0,1)$, and  denote $b=\mathrm{VaR}_{\alpha}(X)$. The optimal loss-dependent contract $I^*$ over $\mathcal{A}_{\textup{IC}}$ is of the following parametric form:
    \begin{align}\label{eq:Istar}
        I^*(x)=(x-d_1)^+-(x-b)^++(x-d_2)^+-(x-(d_2+\pi(I^*)))^+,
    \end{align}
    in which the two parameters, $d_1$ and $d_2$, satisfy $0\le d_1\le b\le d_2$.
\end{theorem}
\begin{proof}
For any $I \in \mathcal{A}_{\textup{IC}}$, following a similar argument as in the proof of Theorem 4.3,  setting $s_1 = I(b) - \pi(I)$ (equivalently, $x_1=b$) and $s_2 = I(b)$, we deduce that the optimal indemnity $I^*$ has the form:
\begin{align*}
    I_1(x; l_1^{(1)}, l_2^{(1)}) &= \left(x - l_1^{(1)} \right)^+ - \left(x - l_1^{(1)} - I(b) \right)^+ \\
    &\quad + \left(x - l_2^{(1)} - I(b) \right)^+ - \left(x - l_2^{(1)} - I(b) - \pi(I_1) \right)^+,
\end{align*}
with $0 \le l_1^{(1)} \le b - I(b) \le l_2^{(1)}$. The function $I_1$ is default-free with the policy limit $I_1(b)+\pi(I_1)$. Let $\pi(I_1) = a_1$; we separate the analysis into two cases.

\textbf{Case 1}: $(1+\eta)\Eb[I_1(X; b - I(b), b - I(b))] > a_1$. Because $(1+\eta)\Eb[I_1(X; b - I(b), l_2^{(1)})] \le a_1$, there exists some $l_2^{(2)}$ satisfying $b - I(b) \le l_2^{(2)} \le l_2^{(1)}$ such that
\(
(1+\eta)\Eb[I_1(X; b - I(b), l_2^{(2)})] = a_1.
\)
Note that 
\begin{align*}
 I_1(x; b \!-\! I(b), l_2^{(2)}) = (x - (b - I(b)))^+ - (x - b)^+ + (x - (l_2^{(2)} + I(b)))^+ - (x - (l_2^{(2)} + I(b) + a_1))^+,
\end{align*}
which matches the form in \eqref{eq:Istar}. One can verify that $I_1(x; b - I(b), l_2^{(2)})$ up-crosses $I_1(x; l_1^{(1)}, l_2^{(1)})$, which in turn implies 
\(
\Eb[u(W(I_1(X; b - I(b), l_2^{(2)})))] \ge \Eb[u(W(I_1(X; l_1^{(1)}, l_2^{(1)})))].
\)

\textbf{Case 2}: $(1+\eta)\Eb[I_1(X; b - I(b), b - I(b))] \le a_1$. Define
\begin{align*}
    I_2(x; c) = (x - (b - c))^+ - (x - (b + a_1))^+,
\end{align*}
for all $c \in [I(b), b]$; observe that the above $I_2(x; c)$ is again consistent with \eqref{eq:Istar}. In addition, we have $I_2(x; I(b)) = I_1(x; b - I(b), b - I(b))$ and $I_2(x; b) \ge I_1(x; l_1^{(1)}, l_2^{(1)})$. Hence, there exists some $c$ over $[I(b), b]$ such that $(1+\eta)\Eb[I_2(X; c)] = a_1$. Moreover, $I_2(x; c)$ up-crosses $I_1(x; l_1^{(1)}, l_2^{(1)})$, proving 
\(
\Eb[u(W(I_2(X; c)))] \ge \Eb[u(W(I_1(X; l_1^{(1)}, l_2^{(1)})))].
\)

By combing the results from the above two cases, we conclude that the optimal loss-dependent contract must be in the form of \eqref{eq:Istar}.
\end{proof}

\end{document}